\documentclass{vldb}

\usepackage{float} 
\usepackage[colorlinks=false, pdfborder={0 0 0}, pdfpagelabels=false]{hyperref} 
\usepackage{flushend}
\usepackage{tabularx}

\hypersetup{
pdfinfo={
    Title={},
    Subject={},
    Author={},
    Creator={},
    Producer={},
    CreationDate={ },
    ModDate={ },
} }

\usepackage{amssymb}
\usepackage{relsize}
\usepackage{lmodern} 
\selectfont
\usepackage{xfrac}

\usepackage{tikz}
\usetikzlibrary{arrows,trees,positioning,shapes.geometric}
\usetikzlibrary{plotmarks}
\definecolor{darkpink}{rgb}{0.91, 0.33, 0.5}
\definecolor{bronze}{rgb}{0.8, 0.5, 0.2}
\tikzset {
edgeannot/.style={auto,minimum size=4mm},
cluster/.style={draw,circle,thick},
edge/.style={thick},
special/.style={draw,shape=regular polygon,regular polygon sides=3,thick},
stext/.style={minimum size=0,inner sep=1mm},
}
\usepackage[algo2e,vlined,ruled,linesnumbered]{algorithm2e}
\SetEndCharOfAlgoLine{}
\SetCommentSty{textit}

\definecolor{algoshade}{HTML}{dddddd}
\def\HiLi{\leavevmode\rlap{\hbox to \hsize{\color{algoshade}\leaders\hrule height .8\baselineskip depth .5ex\hfill}}}

\newlength{\pullupfigure}
\newlength{\pullupcaption}
\usepackage{pgfplots}
\usepackage{pgfplotstable}
\pgfplotsset{compat=1.8}

\pgfplotsset {
graph/.style={
height=46mm,
width=85mm,
legend style={draw=none, fill=none, font=\scriptsize},
  ticklabel style={font=\tiny},
xlabel near ticks,
ylabel near ticks,
scaled ticks=false,
},
graphright/.style={
graph,
},
smallgraph/.style={
graph,
height=48mm,
width=58mm,
},
}

\newcommand{\PP}{\mathbf{P}}
\newcommand{\EE}{\mathbf{E}}

\newcommand{\K}{\mathbf{K}}

\newcommand{\fscore}[1]{\ensuremath{\texttt{score}(#1)}}

\newcommand{\fcost}[1]{\ensuremath{\texttt{cost}(#1)}}

\newcommand{\fcat}[1]{\ensuremath{\texttt{cat}(#1)}}
\newcommand{\pot}[1]{\ensuremath{\texttt{pot}(#1)}}
\newcommand{\extra}[1]{\ensuremath{\texttt{extra}(#1)}}

\newcommand{\pp}{\smaller{\raise.13em\hbox{\text{++}}}}
\newcommand{\lmax}{\ensuremath{l_{\textrm{max}}}}

\newcommand{\tmax}{\ensuremath{t_{\textrm{max}}}}

\newcommand{\edgecost}[2]{\ensuremath{c(#1,#2)}}
\newcommand{\cost}[1]{\ensuremath{c(#1)}}

\newcommand\ifmonospace{\ifdim\fontdimen3\font=0pt }
\newcommand\Cpp{%
\ifmonospace%
    C++%
\else%
    C\kern-.1667em\raise.30ex\hbox{\smaller{++}}%
\fi%
\spacefactor1000 }

\makeatletter
\newcommand*{\rom}[1]{\ensuremath{\mathsmaller{\expandafter\@slowromancap\romannumeral #1@}}}
\makeatother

\newtheorem{definition}{Definition}
\newtheorem{example}{Example}
\newtheorem{lemma}{Lemma}

\begin{document}

\title{Solving Orienteering with Category Constraints\\Using Prioritized Search}

\numberofauthors{2}

\author{
\alignauthor
Paolo Bolzoni\\
        \affaddr{Free University of Bozen-Bolzano}\\
        \affaddr{Faculty of Computer Science}\\
        \affaddr{Piazza Domenicani 3}\\
        \affaddr{39100 Bolzano; Italy}\\
        \email{paolo.bolzoni@stud-inf.unibz.it}
\alignauthor
Sven Helmer\\
        \affaddr{Free University of Bozen-Bolzano}\\
        \affaddr{Faculty of Computer Science}\\
        \affaddr{Piazza Domenicani 3}\\
        \affaddr{39100 Bolzano; Italy}\\
        \email{shelmer@inf.unibz.it}
}
\additionalauthors{}

\maketitle

\begin{abstract}
We develop an approach for solving rooted orienteering problems with category
constraints as found in tourist trip planning and logistics. It is based on
expanding partial solutions in a systematic way, prioritizing promising ones,
which reduces the search space we have to traverse during the search. The
category constraints help in reducing the space we have to explore even
further. We implement an algorithm that computes the optimal solution and also
illustrate how our approach can be turned into an approximation algorithm,
yielding much faster run times and guaranteeing lower bounds on the quality of
the solution found. We demonstrate the effectiveness of our algorithms by
comparing them to the state-of-the-art approach and an optimal algorithm based on
dynamic programming, showing that our technique clearly outperforms these
methods.
\end{abstract}

\section{Introduction}

Imagine a user arriving at the train station of a city, depicted on the
left-hand side of Figure~\ref{fig:ex2}, and needing to be at the airport,
located on the right-hand side of the figure, five hours later. They do not
want to immediately go to the airport, but have a look at the city first. If
the user takes the shortest route from the train station to the airport,
represented by the solid line, they are only able to see a basilica on the way
and still have a lot of spare time when arriving at the airport. According to
the ratings of a tourist guide, the basilica, a pagoda, and a cathedral are the
top points of interest (POIs) that can be visited en route to the airport
within five hours (see dashed line). While this makes much better use of the
time, the user may not be in the mood to visit these POIs, as they may be tired
from traveling and would like a more relaxing route. In this instance, the
dotted line, connecting a park, a ferris wheel, and a statue, is much more
appropriate. As this example shows, the goal is to find a trip with the best
points of interest that makes good use of the available time while at the same
time considering the preferences of a user.

\begin{figure}[hbt]
\centering
\begin{tikzpicture}
[every node/.style={circle},
on grid=true,
inner sep=0pt,
x=1.6cm,y=1.6cm]
\tikzset {
}


\pgfdeclareimage[height=1cm]{start}      {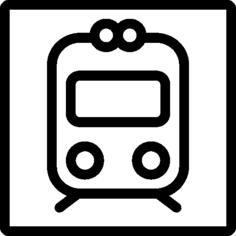}
\pgfdeclareimage[height=1cm]{destination}{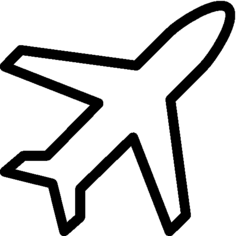}

\pgfdeclareimage[height=1cm]{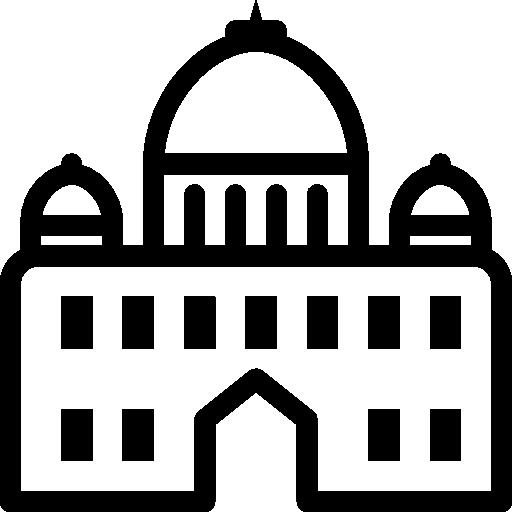}      {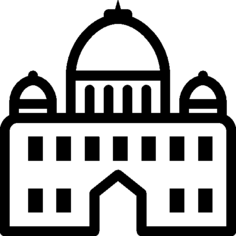}
\pgfdeclareimage[height=1cm]{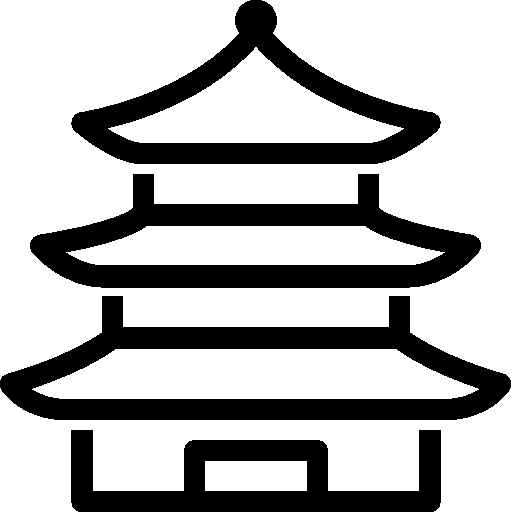}      {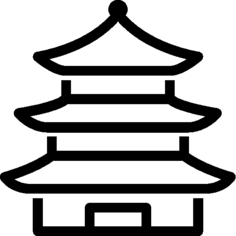}
\pgfdeclareimage[height=1cm]{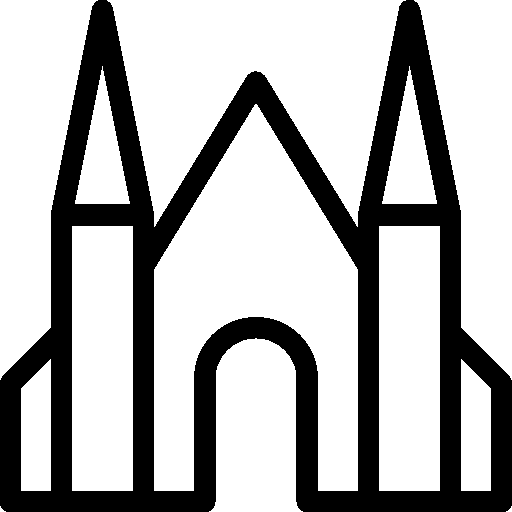}      {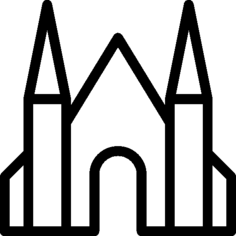}

\pgfdeclareimage[height=1cm]{user1}      {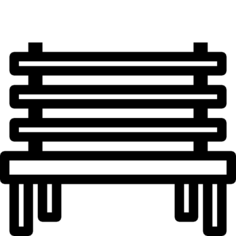}
\pgfdeclareimage[height=1cm]{user2}      {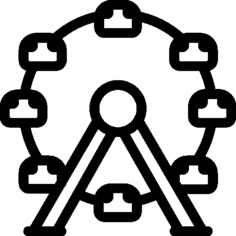}
\pgfdeclareimage[height=1cm]{user3}      {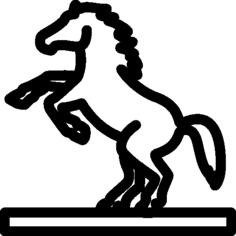}

\pgfdeclareimage[height=1cm]{unused1}      {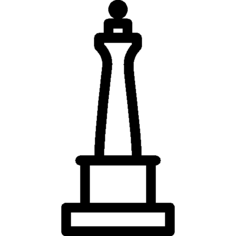}
\pgfdeclareimage[height=1cm]{unused2}      {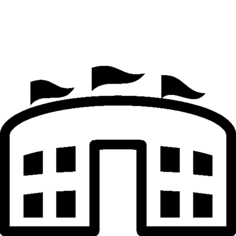}

\node (start) at (0,0) {\pgfuseimage{start}};
\node (destination) at (3,0) {\pgfuseimage{destination}};

\node (best1) at (1.2,0) {\pgfuseimage{best1}};
\node (best2) at (1.8,-1.1) {\pgfuseimage{best2}};
\node (best3) at (3,-1.1) {\pgfuseimage{best3}};

\node (user1) at (1,1.7) {\pgfuseimage{user1}};
\node (user2) at (1.9,1.1) {\pgfuseimage{user2}};
\node (user3) at (3,1.5)  {\pgfuseimage{user3}};

\node (unused1) at (0.7,0.8) {\pgfuseimage{unused1}};
\node (unused2) at (0.2,-0.7) {\pgfuseimage{unused2}};

\draw[very thick] (start.5) -- (best1.175) (best1) -- (destination);
\draw[very thick, dotted] (start) |- (user1) -| (user2.120) (user2.22) |- (user3) -- (destination);
\draw[very thick, densely dashed] (start.355) -- (best1.185) (best1) |-  (best2) -- (best3) -- (destination);

\end{tikzpicture}
\vspace*{\pullupcaption}
\caption{An example with three paths from the subway to the airport: fastest
(solid), suggested by guides (dashed), best for user (dotted).}
\label{fig:ex2}
\vspace*{\pullupfigure}
\end{figure}

Finding itineraries for domains such as tourist trip planning and logistics
often involves solving an orienteering problem. This is because those tasks are
not about determining the shortest path, but the most attractive or the one
covering the most needy customers while satisfying a strict time constraint. We
focus on a variant that assumes that every point of interest (POI) or customer
has a category. This categorization helps a user in expressing preferences,
e.g. a tourist may only be interested in certain types of venues, such as
museums, galleries, and cafes, while certain vehicles may only be able to
serve particular customers.

In general, orienteering is an NP-hard problem, and adding categories does not
change this fact \cite{BHW14}.  We propose an approach based on a best-first
strategy to explore the search space, meaning that we first expand the partial
solutions that show the greatest potential.
We do so with the help of a function approximating the attractiveness of POIs
that can still be added.  Similar to admissible heuristics in an A$^*$-search,
this function needs to satisfy certain properties.
Additionally, we are able to prune partial solutions that cannot possibly
result in an optimal route.

Even though this technique will speed up the search for an optimal solution by
pruning unpromising partial solutions, in the worst case it still has an
exponential run time. Therefore, we describe how to transform our method into
much more efficient approximation algorithms, creating different variants with
important properties concerning the quality of the generated solution and the
run time. In summary, we make the following contributions:

\begin{itemize}

\item We show how to apply a best-first strategy to the problem of
  orienteering with category constraints.

\item Additionally, we turn the optimal algorithm into different approximation algorithms
  proving lower bounds for the quality of a solution and upper bounds for the
  run time.

\item In an experimental evaluation we compare our approach to
  state-of-the-art algorithms, demonstrating its effectiveness. For some
  scenarios we improved the quality of the solutions by about 10\% with a run
  time of just five seconds (with the main competitor taking up to 50 seconds
  and producing worse solutions).

\end{itemize}

The remainder of the paper is organized as follows.  In the next section we
cover related work and in Section~\ref{sec:formal} we formalize the problem.
Sections~\ref{sec:algorithm} and \ref{sec:approx} contain a detailed
description of our approach and an approximation algorithm, respectively,
while in Section~\ref{sec:theobounds} theoretical bounds are provided. This is
followed by an experimental evaluation, comparing our algorithm to the
state-of-the-art algorithm, a simple greedy heuristic, and an optimal
one. Finally, Section~\ref{sec:concl} concludes the paper.

\section{Related Work}
\label{sec:relwork}

A common technique for modeling tourist trip planning~\cite{GKMP14} and
logistics~\cite{Dunn92} problems is to map them to the {\em orienteering
  problem} (OP).  Introduced by Tsiligrides in~\cite{Tsiligrides84}, OP is
about determining a path from a starting node to an ending node in an
edge-weighted graph with a score for each node, maximizing the total score
while staying within a certain time budget. Orienteering is an NP-hard problem
and algorithms computing exact solutions using branch and
bound~\cite{GeLaSe98b,RaYoKa92} as well as dynamic programming
techniques~\cite{LuLiTs11,RiSa09} are of limited use, as they can only solve
small problem instances.  Consequently, there is a body of work on
approximation algorithms and heuristics, most of them employing a two-step
approach of partial path construction~\cite{Kell89,Tsiligrides84} and
(partial) path improvement~\cite{BCKLMM07,ChKoPa08,SeSe06}. Meta-heuristics,
such as genetic algorithms~\cite{TaSm00}, neural networks~\cite{WSGJ95}, and
ant colony optimization~\cite{LiKuSm02} have also been tested. However, none
of the approaches investigate OP generalized with categories.  For a recent
overview on orienteering algorithms, which still mentions orienteering with
category constraints as an open problem, see~\cite{GKMP14}.

There is also work on planning and optimizing errands, e.g., someone wants to
drop by an ATM, a gas station, and a pharmacy on the way home. The generalized
traveling salesman version minimizes the time spent on this trip
\cite{RiTs13}, while the generalized orienteering version maximizes the number
of visited points of interest (POIs) given a fixed time budget. However, as
there are no scores, no trade-offs between scores and distances are
considered.

Another, different take on tourist trip planning is to look at the user
experience. For instance, Maruyama et al. investigate feedback elicitation
from users in a tourist trip planning context \cite{MSMY04b}, while Castillo et
al. look at context-based recommendation and collaborative filtering
\cite{samap08}. However, most of these approaches do not consider the actual
route planning or use very simple greedy heuristics.

Adapting an existing algorithm for OP would be a natural starting point for
developing an approximation algorithm considering categories.  However, many
of the existing algorithms have a high-order polynomial complexity or no
implementation exists, due to their very complicated structure. Two of the
most promising approaches we found were the segment-partition-based technique
by Blum et al.~\cite{BCKLMM07} and the method by Chekuri and P\'{a}l,
exploiting properties of submodular functions~\cite{ChPa05}. The latter
approach, a quasi-polynomial algorithm, is still too slow for practical
purposes. Additionally, common to all of the approaches, though, is breaking
down the itinerary recursively into smaller and smaller segments, which get
assembled into a complete tour. If we just run these algorithms without any
alterations on POIs with categories, it is very likely that the solution
violates the max-n type constraints. A fix would be to try out all possible
distributions of max-n type constraints for every recursive call. For example,
given a max-n type constraint of 3 for a category and assuming \texttt{recl}
computes the left half of a route and \texttt{recr} the right one, we would
have to make the following calls with max-n type constraints:
\texttt{recl}(0), \texttt{recr}(3); \texttt{recl}(1), \texttt{recr}(2);
\texttt{recl}(2), \texttt{recr}(1); \texttt{recl}(3), \texttt{recr}(0). While
this would guarantee an answer respecting the max-n type constraints, it would
also blow up the computational costs.

Nevertheless, Singh et al. \cite{SKGKB07} modified the algorithm of Chekuri
and P\'{a}l by introducing spatial decomposition for Euclidean spaces in the
form of a grid, making it more efficient.  Our work on CLuster Itinerary
Planning (CLIP) \cite{BHW14} for OPs on road networks with category
constraints was inspired by~\cite{SKGKB07}.

\section{Problem Formalization}
\label{sec:formal}

We assume a set of points of interest (POIs) $p_i, 1\leq i\leq n$, represented
by $\PP$.  The POIs, together with a starting and a destination node, denoted
by $s$ and $d$, respectively,  are connected by a complete, metric, weighted,
undirected graph $G=(\PP\cup\{s,d\},\EE)$, whose edges, $e_l \in \EE = \{(x,y)
\mid x,y \in \PP \cup\{s,d\}\}$ connect them.  Each edge $e_l$ has a cost
$\edgecost{p_i}{p_j}$ that signifies the duration of the trip from $p_i$ to
$p_j$, while every node  $p_i\in\PP$ has a cost $\cost{p_i}$ that denotes its
visiting time.  Each POI belongs to a certain category, such as \emph{museums,
restaurants}, or \emph{galleries}.  The set of $m$ categories is denoted by
$\K$ and each POI $p_i$ belongs to exactly one category $k_j, 1\leq j\leq m$.
Given a $p_i$, $\fcat{p_i}$ denotes the category $p_i$ belongs to and
$\fscore{p_i}$ its score or reward, with higher values indicating higher
interest to the user. Finally, users have a certain maximum time in their
budget to complete the itinerary, denoted by \tmax{}.

\begin{definition}\label{def1}
  \emph{(Itinerary)} An itinerary $\mathcal{I}$ starts from a starting point $s$ and
  finishes at a destination point $d$ ($s$ and $d$ can be identical).
  It includes an ordered sequence of connected nodes $\mathcal{I}=\langle s,
  p_{i_1},$ $p_{i_2}, \ldots,$ $p_{i_q}, d\rangle$, each of which is visited once. We
  define the \emph{cost} of itinerary $\mathcal{I}$ to be the total duration of
  the path from $s$ to $d$ passing through and visiting the POIs in
  $\mathcal{I}$, $\fcost{\mathcal{I}}=\edgecost{s}{p_{i_1}} + \cost{p_{i_1}} +
  \sum_{j=2}^{q}(\edgecost{p_{i_{j-1}}}{p_{i_j}}+\cost{p_{i_j}})+\edgecost{p_{i_q}}{d}$,
  and its {\emph score} to be the sum of the scores of the individual POIs
  visited, $\fscore{\mathcal{I}}=\sum_{j=1}^{q} \fscore{p_{i_j}}$.
\end{definition}

\begin{figure}[htb]
\vspace*{-.4cm}
\centering
\scalebox{.86}{
\begin{tikzpicture}
\tikzset {
inner sep=0,minimum size=9mm,
}
\node[cluster,densely dashed] (legend k1) at (-32mm, 6mm) {$p_i$};
\node[cluster] (legend k2) [below=7mm of legend k1] {$p_j$};

\node[stext,right] (legend k1 text) at (legend k1.east) {{\scriptsize $\fscore{p_i} = 0.9$}};
\node[stext] [above=-1mm of legend k1 text.north] {{\scriptsize $\fcat{p_i} = k_1$}};
\node[stext] [below=-1mm of legend k1 text.south] {{\scriptsize $i = 1,3$}};

\node[stext,right] (legend k2 text) at (legend k2.east) {{\scriptsize $\fscore{p_j} = 0.5$}};
\node[stext] [above=-1mm of legend k2 text.north] {{\scriptsize $\fcat{p_j} = k_2$}};
\node[stext] [below=-1mm of legend k2 text.south] {{\scriptsize $j = 2,4$}};

\node[special] (s) {$s$};
\node[special] (d) [right=1.7cm of s] {$d$};
\node[cluster,densely dashed] (p1) [below right=1cm of s]{$p_1$};
\node[cluster] (p2) [above right=1cm of s]{$p_2$};
\node[cluster,densely dashed] (p3) [above right=1cm of d]{$p_3$};
\node[cluster] (p4) [below right=1cm of d]{$p_4$};

\node[stext,below] at (p1.south) {$\cost{p_1} = 1$};
\node[stext,above] at (p2.north) {$\cost{p_2} = 1$};
\node[stext,above] at (p3.north) {$\cost{p_3} = 1$};
\node[stext,below] at (p4.south) {$\cost{p_4} = 1$};

\draw [edge] (s) to node[edgeannot,swap] {4} (p1);
\draw [edge] (s) to node[edgeannot] {2}  (p2);
\draw [edge] (p1) to node[edgeannot,swap] {6} (d);
\draw [edge] (p2) to node[edgeannot] {2} (p3);
\draw [edge] (p3) to node[edgeannot,swap] {3} (d);
\draw [edge] (p3) to node[edgeannot] {2} (p4);
\draw [edge] (p4) to node[edgeannot,swap] {1} (d);
\end{tikzpicture}
}

{\vspace{-2ex}}
\vspace*{\pullupcaption}
\caption{Itinerary including $n=4$ POIs}
\label{fig:ex1}
\vspace*{\pullupfigure}
\end{figure}
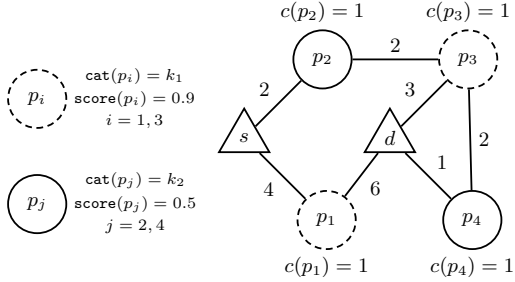

\begin{example}

\label{ex:example1} Figure~\ref{fig:ex1} shows an example with four POIs,
$p_1,p_2,p_3$, and $p_4$, along with their distances, visiting times, scores,
and categories. We simplify the graph slightly to keep it readable: all POIs of
the same category have the same score we also omit some edges.  One example
itinerary between $s$ and $d$ is the one that includes only $p_1$, {\emph
i.e.,} $\mathcal{I}_1=\langle s,p_1,d\rangle $, while a second one includes
$p_2$ and $p_3$, {\emph i.e.,} $\mathcal{I}_2=\langle s,p_2,p_3,d\rangle$, a
third itinerary can be $\mathcal{I}_3=\langle s, p_2, p_3, p_4, d \rangle$.
Their costs and scores are as follows:

\begin{itemize}
\item $\mathcal{I}_1=\langle s,p_1,d\rangle $: \newline
  $\fcost{\mathcal{I}_1}=4+1+6=11$, \newline
  $\fscore{\mathcal{I}_1}=0.9$;

\item $\mathcal{I}_2=\langle s,p_2,p_3,d\rangle $: \newline
  $\fcost{\mathcal{I}_2}=2+1+2+1+3=9$, \newline
  $\fscore{\mathcal{I}_2}=0.5+0.9=1.4$;

\item   $\mathcal{I}_3=\langle s,p_2,p_3,p_4,d\rangle $: \newline
  $\fcost{\mathcal{I}_3}=2+1+2+1+2+1+1=10$, \newline
  $\fscore{\mathcal{I}_3}=0.5+0.9+0.5=1.9$.
\end{itemize}
\end{example}

In general, given a subset $P$ of $\PP$, our goal is to place them in an
itinerary from $s$ to $d$ as defined above. Given the traveling and visiting
times as well as the scores, we need to build an itinerary with duration
smaller than \tmax{} and maximum cumulative score.  As already mentioned, we
introduce an additional constraint specifying the number of POIs per category
that can be included in the final itinerary.  More precisely, we introduce a
parameter $\max_{k_j}$ for each category $k_j$ that is set by the user to the
maximum number of POIs in a category that he or she prefers to visit during the
trip.  We are now ready to define the \emph{Orienteering Problem with Maximum
Point Categories} (\emph{OPMPC}).

\begin{definition}
  \vskip -1ex
  \label{def:opmpc}
  \emph{(OPMPC)} Given a starting point $s$, a destination point $d$, $n$ points of
  interest $p_i \in \PP$, with scores described by the function
  \fscore{p_i}, visiting times $\cost{p_i}, 1\leq i\leq n$, traveling
  times $\edgecost{x}{y}$ for $x,y \in \PP \cup \{s,d\}$, categories
  $k_j \in \K, 1\leq j\leq m$, and the following two parameters: (a)
  the maximum total time \tmax{} a user can spend on the itinerary
  and, (b) the maximum number of POIs $\max_{k_j}$ that can be used
  for the category $k_j$ $(1\leq j \leq m)$, a solution to the OPMPC
  is an itinerary $\mathcal{I}=\langle
  s,p_{i_1},p_{i_2},\ldots,p_{i_q},d\rangle $, $1 \leq q \leq n$, such
  that
\begin{itemize}
  \item the total score of the points, $\fscore{\mathcal{I}}$, is
    maximized;
  \item no more than $\max_{k_j}$ POIs are used for category $k_j$;
  \item the time constraint is met, {\emph i.e.,}
    $\fcost{\mathcal{I}}\leq \tmax{}$.
  \end{itemize}
\end{definition}

\begin{example}
  In the presence of categories $k_1$ with $\max_{k_1}=1$ and $k_2$ with
  $\max_{k_2}=1$, and assuming that $\tmax{}=10$, we can observe the following
  about the itineraries in Example~\ref{ex:example1}: Itinerary $\mathcal{I}_1$
  is infeasible since its cost is greater than \tmax{}, while the other two
  fulfill the time requirement. Comparing $\mathcal{I}_2$ and $\mathcal{I}_3$,
  we can see that $\mathcal{I}_3$ is of higher benefit to the user, even though
  it takes more time to travel between $s$ and $d$. However, it cannot be
  chosen since it contains two POIs from $k_2$. Itinerary $\mathcal{I}_2$
  contains two POIs, each from a different category and it could be one
  recommended to the user.
\end{example}

In Table~\ref{tab:symbols} we give a summary of the notation we use in this paper.

\begin{table}[ht]
\begin{tabular}{ll}
\hline
$p_i$                &  a single POI           \\
$\PP$                &  the set of all POIs \\
$n$                  &  number of POIs         \\
$k_j$                &  a single category \\
$\K$                 &  the set of all categories \\
$m$                  &  number of categories   \\
$\max_{k_j}$          &  category constraint for category $k_j$ \\
$\fcat{p_i}$         &  category of POI $p_i$  \\
$\fscore{p_i}$       &  score of POI $p_i$     \\
$\cost{p_i}$          &  cost of $p_i$ (visiting time) \\
$\mathtt{utility}(p_i)$ & utility of $p_i$ (used for greedy strategy) \\
$G$                  &  graph containing POIs as nodes \\
                     & (basically a road network) \\
$c(p_i, p_j)$        &  cost of an edge        \\
${\cal I}$           &  itinerary \\
$\tmax$              &  time constraint \\
$s$                &  starting node \\
$d$                &  destination node \\
$u_{k_j}$            &  no of POIs of category $k_j$ in the itinerary \\
$\PP_{rem(\cal I)}$  &  set of POIs $\PP$ without POIs in $\cal I$ \\
\fcost{\cal I} & total cost of $\cal I$ (traveling and visiting times) \\
\fscore{\cal I}      &  current score of an itinerary \\
\pot{\cal I}         &  potential score of an itinerary \\
                     &  (upper bound on achievable score) \\
\extra{\cal I}       &  upper bound on score of still reachable POIs \\
                     & in $\PP_{rem(\cal I)}$\\
\hline
\end{tabular}
\vspace*{\pullupcaption}
\caption{Notation}
\label{tab:symbols}
\vspace*{\pullupfigure}
\end{table}

\section{Our Approach}
\label{sec:algorithm}

We roughly follow a best-first strategy, meaning that we keep all solutions
generated so far in a priority queue sorted by their potential score.  A
solution is represented by the set of POIs it contains. In each step, we take
the solution with the highest potential score from the queue, expand it with
POIs that have not been visited yet, and re-insert the expanded solutions back
into the queue. An important difference to the classic best-first search is
that in our case it is not straightforward to identify the goal state, as we do
not know the score of the optimal itinerary a priori.  Consequently, we have to
run the algorithm until there are no solutions left in the queue that have a
higher potential score than the best found so far (which then becomes the
result we return).  Applying this approach in a straighforward fashion is not
practical in most cases as OPMPC is an intractable problem.  Therefore, after
presenting an algorithm that computes the optimal solution, we turn it into an
approximation algorithm.

\subsection{Potential Score}

An important aspect of our search strategy is the computation of the potential
score of a partial solution, which has to be an upper bound of the score of the
fully expanded solution based on this partial solution. Basically, this follows
the principle of admissible heuristics found in the A$^*$-algorithm, with the
difference that we never underestimate the value, since we consider scores and
not costs.

\begin{definition}\label{defpot}
  \emph{(Potential score)} Given a (partial) itinerary ${\cal I} = \langle s, 
p_1, p_2, \dots, p_i, d \rangle$, its potential score $\pot{\cal I}$ is defined as follows:
\begin{eqnarray*}
\pot{\cal I} = \fscore{\cal I} + \extra{\cal I}
\end{eqnarray*}
\noindent
where $\fscore{{\cal I}} = \sum_{j=1}^{i} \fscore{p_j}$ and $\extra{\cal I}$ is a heuristic
never underestimating the additional score that is still possible for $\cal I$.
\end{definition}

Algorithm~\ref{algo:extra} shows the pseudocode for computing $\extra{\cal I}$.
We obtain an approximation (without underestimating the true extra score we can
still get) by figuring out how many of the top-scoring POIs we could
potentially still fit into the partial itinerary $\cal I$.  In a first step, we
remove all POIs that are already in $\cal I$ from the set of all POIs (line
\ref{alg:rem}).  We call this set of remaining POIs $\PP_{rem(\cal I)}$. Then
we discard all POIs in $\PP_{rem(\cal I)}$ that are too costly to travel to
from the penultimate POI $p_i$ to $d$ (line \ref{alg:costly}).  Now that we
have reduced the set of POIs, we turn to the category constraints to get an
even more accurate picture (line \ref{alg:cat}). For every category we know how
many POIs we can still fit into the itinerary without violating the category
constraints. As previously defined, $\max_{k_j}$ is the maximum number of POIs
that can be chosen for category $k_j$. Let $u_{k_j}$ be the number of POIs of
category $k_j$ currently found in $\cal I$ (line \ref{alg:ukj}), then we know
that we can only fit $(\max_{k_j} - u_{k_j})$ more POIs of category $k_j$ into
$\cal I$ (line \ref{alg:maxminusukj}).  For each category, we sort the POIs in
$\PP_{rem(\cal I)}$ according to their score in descending order and add up the
scores of the top $(\max_{k_j} - u_{k_j})$ POIs of each category (lines
\ref{alg:addstart} to \ref{alg:addstop}). This comprises $\extra{\cal I}$.

\begin{algorithm2e}[htb]
\caption{ $e \gets \mathtt{Extra}(q, {\cal I}, G, \PP)$ }
\label{algo:extra}

\KwIn{query $q$ (consisting of $s$, $d$, $\tmax$, and category constraints), an itinerary ${\cal I}$,
graph $G$, set $\PP$ of all POIs} 
\KwOut{potential extra score $\extra{\cal I}$}

\SetKw{kContinue}{continue}
\SetKw{kNot}{not}
\SetKw{kAnd}{and}
\SetKw{kBreak}{break}
\SetKw{kFrom}{from}
\SetKw{kTo}{to}

$x \gets 0$ \;
$\PP_{rem(\cal I)} \gets \PP \setminus {\cal I}$ \; \label{alg:rem}
$p_i \gets \mbox{last but one node in } {\cal I}$ \;

\ForEach {$p \in \PP_{rem(\cal I)}$} { \label{alg:costly}
  \lIf {$\fcost{{\cal I}} - \edgecost{p_i}{d} + \edgecost{p_i}{p} + \cost{p} +
    \edgecost{p}{d} > \tmax$} {
    $\PP_{rem(\cal I)} \gets \PP_{rem(\cal I)} \setminus \{ p \}$ 
  }
}

\ForEach {$\mbox{category } k_j \in q$} { \label{alg:cat}
  $u_{k_j} \gets | \{ p \in {\cal I} | \fcat{p} = k_j \} |$ \; \label{alg:ukj}
  \If {$\max_{k_j} > u_{k_j}$} { \label{alg:maxminusukj} 
    $\PP_{rem(\cal I)}^{k_j} \gets \{ p \in \PP_{rem(\cal I)} | \fcat{p} =
    k_j \}$ \;
    $\mbox{sort } \PP_{rem(\cal I)}^{k_j} \mbox{ in descending order of
      scores}$ \; \label{alg:addstart}
    \For {$i\ \kFrom\ 1\ \kTo\ {\max_{k_j}-}u_{k_j} $ } {
      $p_i \gets i \mbox{th POI in } \PP_{rem(\cal I)}^{k_j}$ \;
      $x \gets x + \fscore{p_i}$ \; \label{alg:addstop}
    }
  }
}

\Return $x$
\end{algorithm2e}

\subsection{Our Algorithm}

We are now ready to describe the pseudocode in Algorithm~\ref{algo:mainloop}.
For the moment disregard the parts marked in gray: these are optimizations that
will be discussed in the following section. Partial solutions are kept in a
double-ended priority queue (also called deque). The deque is initialized with
the empty solution $e$ just containing $s$ and $d$ and a score of 0. We take
the solution with the highest potential score out of the deque and expand it,
one by one, with the POIs that have not been visited yet. In order to figure
out whether a solution violates the time constraint $\tmax$, we have to compute
the (shortest) length of the itinerary of a solution. As every partial solution
contains only a limited number of POIs, we compute this length in a brute-force
way.  We discard any expanded solutions that violate the time constraint or any
of the category constraints and put all valid expansions back into the deque.
If an expanded solution is better than the current best solution, we update it.
We continue until the potential score of the solution taken from the deque is
smaller than the best found so far (which then becomes the answer we return).

\begin{algorithm2e}[htb]
\caption{ $b \gets \mathtt{MainLoop}(q, G)$ }
\label{algo:mainloop}

\KwIn{query $q$ (consisting of $s$, $d$, $\tmax$, and category constraints), 
graph $G$, set $\PP$ of all POIs} 
\KwOut{best solution $b$}

\SetKw{kContinue}{continue}
\SetKw{kNot}{not}
\SetKw{kAnd}{and}
\SetKw{kBreak}{break}

$Q$ priority deque for partial solutions\;

$e.\mathtt{cost} \gets c(q.s,q.d)$ \;  
$e.\mathtt{extra} \gets \mathtt{Extra}(q, e, G, \PP)$ \;

\HiLi $b \gets \mathtt{ExtendGreedily}(q, e)$ \; \label{alg:firstgreedy}
$Q.\mathtt{push}(e)$ \;


\While{\kNot $Q.\mathtt{empty}()$} {
    $s \gets Q.\mathtt{popMaximum}$() \; \label{alg:timeoutexitb}
    \ForEach{$z \in \mathtt{Expand}(s, G)$ } {  \label{alg:timeoutexite}
        \lIf {$z$ violates any category constraint} {$\kContinue$}
        $z \gets \mathtt{ComputePathAndScore}(z)$  \;
        \lIf {$z.\mathtt{cost} > q.\tmax{}$} {$\kContinue$}
        $z.\mathtt{extra} \gets \mathtt{Extra}(q, z, G, \PP)$ \;
        \lIf {$z.\mathtt{score} > b.\mathtt{score}$} {$b \gets z$ }
        $Q.\mathtt{push}(z)$ \; \label{alg:insidepush}
\HiLi   \If {at an early stage} { \label{alg:earlystageb}
\HiLi       $y \gets \mathtt{ExtendGreedily}(q, z)$ \;
\HiLi    \lIf {$y.\mathtt{score} > b.\mathtt{score}$} {$b \gets y$ } \label{alg:earlystagee}
        }
    }
    \While {$b.\mathtt{score} \newline{} \phantom{iiiii} > Q.\mathtt{minimum}().\mathtt{score} + Q.\mathtt{minimum}().\mathtt{extra}$}  { \label{alg:fastexit}
        $Q.\mathtt{popMinimum}() $ \;
    }
}

\Return $b$
\end{algorithm2e}

\subsection{Further Optimizations}

With the help of the score of the current best solution, we can do some
pruning, because only partial solutions whose potential score is better have to
be considered. In order to be able to prune earlier and more aggressively, we
initialize the best solution found so far with an itinerary found quickly using
a greedy algorithm (see line~\ref{alg:firstgreedy} and
Algorithm~\ref{algo:greedyexpand}). Given a partial itinerary ${\cal I} =
\langle s, p_1, p_2, \dots, p_i, d \rangle$, the greedy strategy adds to this
path a POI $p \in \PP_{rem({\cal I})}$ such that its utility

\begin{eqnarray*}
\mathtt{utility}(p) &=& \frac{\fscore{p}}{\edgecost{p_i}{p} + \cost{p} + \edgecost{p}{d}}
\end{eqnarray*}

\noindent
is maximal and no constraints are violated.  We repeat this until no further
POIs can be added to the itinerary.

\begin{algorithm2e}[htb]

\caption{ $s \gets \mathtt{ExtendGreedily}(q, z)$ }
\label{algo:greedyexpand}

\KwIn{ query $q$, solution $z$ }
\KwOut{ solution $s$ }

\SetKw{kContinue}{continue}
\SetKw{kNot}{not}
\SetKw{kAnd}{and}
\SetKw{kBreak}{break}

$s,s' \gets z$ \;

\While {$\mathtt{length}(s') \leq q.\tmax{}$ } {
    $s \gets s', u \gets -\infty$ \;
    \ForEach {$p \in \PP_{rem(s.{\cal I})}$ } {
        \lIf {$p$ violates any category constraint of $s'$} {$\kContinue$}
        \If { $u < \mathtt{utility}(p)$ } {
            $u  \gets \mathtt{utility}(p), p' \gets p$ \;
        }
    }
    $\mathtt{AppendToEnd}(s', p') $ \;
}

\Return $s$
\end{algorithm2e}

We can also use a greedy approach to improve the current best solution while
the algorithm is running (see lines~\ref{alg:earlystageb}
to~\ref{alg:earlystagee}). When expanding a partial solution taken from the
deque, we also complete it with a greedy strategy. This improves the current
best solution more rapidly, leading to more aggressive pruning and is
especially useful in earlier stages of the algorithm, as there is still a lot
of unexplored search space.  We may want to stop the greedy expansion in later
stages, due to it being less effective, as we have already explored a lot of
promising directions.  The threshold $g \in [0,1]$ defines when the algorithm
is in a \emph{early stage}, it checks each expanded solution $z$: whenever the
ratio $\sfrac{z.\mathtt{cost}}{\tmax} \leq g$, it does a greedy expansion.  We
investigate this parameter in more detail in Section~\ref{sec:benchmark} on the
experimental evaluation.

\section{Approximation Algorithms}
\label{sec:approx}

Even though pruning and further optimizations speed up the search for a
solution, in the worst case we still face exponential
run time. After all, we are looking for an optimal solution for an NP-hard
problem. Here we present variations of Algorithm~\ref{algo:mainloop} that turn
it into an approximation algorithm. In principle, we can guarantee an upper bound
on the run time or a lower bound on the score.

\subsection{Bounding the Score}
\label{sec:approxscore}

In line~\ref{alg:fastexit} of Algorithm~\ref{algo:mainloop} we 
remove (partial) solutions from the end of the deque that have a potential
score smaller than the best solution $b$ found so far. 
Since solutions are sorted by an overestimation of
their score, we never accidentally get rid of a partial solution that could be
expanded into the optimal one.
However, we could prune more aggressively
by introducing a \emph{cut factor} $c$ ($c \geq 1$) and checking whether
$c \cdot b.\mathtt{score} > Q.\mathtt{minimum}().\mathtt{score} 
+ Q.\mathtt{minimum}().\mathtt{extra}$.\footnote{If
we set $c$ to 1, we get the original algorithm, i.e., the approximation
algorithm is actually a generalization.}
The larger $c$, the more partial solutions will be ignored. However, by doing
so we also risk losing the optimal and other good solutions. Nevertheless, 
we are sure to get a solution that
guarantees at least $\alpha = \sfrac{1}{c}$ of the score of the optimal solution.

\subsection{Bounding the Run Time}
\label{sec:approxruntime}

While the approach in the previous section makes sure that we always get a
certain amount of the optimal score, we are not able to bound its run
time. All we know is that it is faster, since we prune more partial
solutions. By modifying Algorithm~\ref{algo:mainloop} in a different way, we
can obtain a run time bound. Instead of allowing the queue to become
arbitrarily long, we limit its length to a maximum of $\lmax$ entries. 
In this way, we put a limit on the number of partial solutions that can be
expanded. Before pushing an expanded solution $z$ back into the deque in 
line~\ref{alg:insidepush}, we check whether there is still space. If there is not enough space,
we remove the partial solution with the lowest potential score (this could
also be $z$).
While we cannot determine the quality of the answer a priori, we
know the value for $\alpha$ upon completion of the algorithm if
we keep track of the largest potential score of all the
partial solutions we discarded. Assume that
$z_{\max}$ has the largest potential score among the discarded solutions. Then
we know that $\alpha = \sfrac{b}{\pot{z_{max}}}$, where $b$ is the optimal solution.

Another, more direct, way of controlling the run time is to set an explicit
limit $r$ for it. After the algorithm uses up the allocated time, it returns
the best solution found so far. For example, we can check between
line~\ref{alg:timeoutexitb} and~\ref{alg:timeoutexite} whether there is any
time left. If not, we jump out of the while-loop.  Again, it is not possible to
prescribe a value for $\alpha$ beforehand, but we can determine its value when
the algorithm finishes: in this case $z_{\max}$ is the partial solution with
the largest potential score that is still in the deque when we stop.

\section{Properties and Bounds}
\label{sec:theobounds}

In the following we prove that the pruning we do leads to an algorithm
computing the itinerary with the maximal score. Additionally, we prove
the bounds for the score and the run time of the approximation algorithms.

\subsection{Correctness of Pruning}

Let ${\cal I}_c$ be a {\em complete itinerary}, i.e., no further POIs can be
added to ${\cal I}_c$ without violating a constraint.
We have to prove that for every complete itinerary ${\cal I}_c$, 
$\pot{{\cal I}_c} = \fscore{{\cal I}_c} \leq \pot{{\cal I}_p}$, where 
${\cal I}_p$ is a partial itinerary of ${\cal I}_c$
(${\cal I}_p \subset {\cal I}_c$). 
${\cal I}_p \subset {\cal I}_c$ iff the POIs $p_1, p_2, \dots, p_p$ in 
${\cal I}_p$ are a subset of the POIs $p_1, p_2, \dots, p_c$ in 
${\cal I}_c$. The potential score of ${\cal I}_c$ is equal to its score, since
$\extra{{\cal I}_c} = 0$ (no more POIs can be added).
 We can actually show a stronger statement that includes the one above.

\begin{lemma}
As we expand solutions, their potential score decreases: given two itineraries 
${\cal I}_1$ and ${\cal I}_2$ with ${\cal I}_1 \subset {\cal I}_2$, it follows
that $\pot{{\cal I}_2} \leq \pot{{\cal I}_1}$.
\end{lemma}

\begin{proof}
We prove the lemma by structural induction.

\begin{itemize}

\item {\bf Induction step ($n \rightarrow n+1$):} \\
  Given two itineraries ${\cal I}_n$ with POIs $p_1, p_2, \dots, p_n$ and 
    ${\cal I}_{n+1}$ with POIs $p_1, p_2, \dots,$ $p_{n+1}$, we have to show that
  $\pot{{\cal I}_{n+1}} \leq \pot{{\cal I}_n}$
\begin{eqnarray*}
& \Leftrightarrow & \fscore{{\cal I}_{n+1}} + \extra{{\cal I}_{n+1}} \\
&                 & \leq \fscore{{\cal I}_n} + \extra{{\cal I}_n} \\
& \Leftrightarrow & \fscore{{\cal I}_n} + \fscore{p_{n+1}} + \extra{{\cal I}_{n+1}} \\
&                 & \leq \fscore{{\cal I}_n} + \extra{{\cal I}_n} \\
& \Leftrightarrow & \fscore{p_{n+1}} + \extra{{\cal I}_{n+1}} \\
&                 &\leq \extra{{\cal I}_n}
\end{eqnarray*}

\begin{sloppypar}
Assuming that we can still add (up to) $k$ POIs to ${\cal I}_n$,
$\extra{{\cal I}_n}$ is computed by taking the top-k POIs from
$\PP_{rem({\cal I}_n)}$, whereas $\extra{{\cal I}_{n+1}}$ is computed by taking
the top-(k-1) POIs from $\PP_{rem({\cal I}_{n+1})}$. We have just expanded 
${\cal I}_n$ by adding $p_{n+1}$ to it, so we know that 
$\PP_{rem({\cal I}_{n+1})} \cup \{ p_{n+1} \} \subseteq 
\PP_{rem({\cal I}_n)}$, as $p_{n+1}$ is missing from 
$\PP_{rem({\cal I}_{n+1})}$ (compared to 
$\PP_{rem({\cal I}_n)}$) and 
$\PP_{rem({\cal I}_{n+1})} \cup \{ p_{n+1} \}$ can at most reach 
$\PP_{rem({\cal I}_n)}$, because the remaining time in itinerary 
${\cal I}_{n+1}$ is less than the remaining time in itinerary ${\cal I}_{n}$. Therefore, the choice we get when picking 
top-scoring POIs we pick for the left-hand side of the inequality is more
restricted compared to the the right-hand side.
\end{sloppypar}

\item {\bf Base case (empty itinerary):} \\
The empty itinerary ${\cal I}_0$ contains no POIs, therefore 
$\pot{{\cal I}_0} = \extra{{\cal I}_0}$ and we choose the top-scoring POIs
from the largest possible set $\PP_{rem({\cal I}_0)}$, which means that
${\cal I}_0$ has the largest possible potential score.
\end{itemize}
\end{proof}

\subsection{Lower Bounding the Score}
\label{sec:lowerboundscore}

Using the approximation technique from Section~\ref{sec:approxscore}, we can
guarantee a lower bound for the score we achieve. At any given time, when
discarding a partial solution ${\cal I}_p$, we know that its potential score
$\pot{{\cal I}_p}$ is at most a factor $c$ greater than the best solution $b$
found so far: $c \cdot b \geq \pot{{\cal I}_p}$. If there is no further change
for $b$, the ratio of achieving $\alpha = \sfrac{1}{c}$ of the best score
holds, as $\pot{{\cal I}_p}$ never underestimates the score of a full
expansion of ${\cal I}_p$. If $b$ is replaced by a new best score $b'$, we
then know that the previously discarded partial solutions are actually closer
to the optimal solution, as $b'$ can only be greater than $b$. In fact, the
cut factor of the previously discarded partial solutions improves to 
$c' = \sfrac{b}{b'} \cdot c$. Additionally, the lemma in the previous section states
that the potential score of an expanded solution can never go up, which means
that we are also on the safe side here: we cannot lose a complete solution
that is less than a factor $\alpha$ of the optimal solution by discarding 
${\cal I}_p$.

Even when using the approximation techniques bounding the run time shown in
Section~\ref{sec:approxruntime}, which do not allow us to set a lower bound
for the score a priori, we can draw conclusions about the factor $\alpha$
after the algorithm finishes. When running our algorithm with a limited queue
length, we keep track of the solution with the largest potential score 
$\pot{{\cal I}_{\max}}$ that we have discarded. When execution stops, we can now
compare $\pot{{\cal I}_{\max}}$ to the best solution $b$ that is returned,
knowing that $\alpha$ has to be at least 
$\sfrac{b}{\pot{{\cal I}_{\max}}}$. Using an explicit time limit $r$, the
largest (implicitly) discarded $\pot{{\cal I}_{\max}}$ can be found at the
head of the queue and we can compute $\alpha$ as indicated above. This even
makes it possible to run our algorithm in an iterative fashion. Once the time
limit $r$ has been reached, but a user is not satisfied with the current level
of quality, they can continue running the algorithm for another limited time
period, checking $\alpha$ again. In this way, the quality of the solution and
the run time can be balanced on the fly.

\subsection{Upper Bounding the Run Time}

When setting an explicit time limit $r$, bounding the run time is
straightforward. In case of a limited queue length, drawing conclusions about
the run time is more complicated.
The worst case for a best-first search algorithm is a scenario in which all
the POIs have the same score, the same cost, and a similar distance from $s$,
$d$ and each other. In this case we first expand the empty partial solution
generating all itineraries of length one. In the next step, we expand all these
itineraries to those of length two, then all of those to those of length
three, and so on, meaning that no pruning is taking place. We now give
an upper bound for the number of generated itineraries and then illustrate the
effect of a bounded queue length on this number.

We have $m$ categories $k_j$, $1 \leq j \leq m$, each with
the constraint $max_{k_j}$ and the set 
$K_{k_j} = \{p_i | p_i \in P, cat(p_i) = k_j\}$ containing
the POIs belonging to this category. The set of relevant
POIs $R$ is equal to $\bigcup_{j=1,max_{k_j} > 0}^m K_{k_j}$,
we denote its cardinality by $|R|$ (assuming that $R$ only contains POIs
actually reachable from $s$ and $d$ in time \tmax{}).
For the first POI we have $|R|$ choices, for the next
one $|R|-1$, and so on until we have created paths of
length $\lambda = \sum_{j=1,max_{k_j} > 0}^m max_{k_j}$. We cannot possibly
have itineraries containing more than $\lambda$ POIs, as this
would mean violating at least one of the category constraints.
So, in a first step this gives us $\prod_{i=0}^{\lambda-1} (|R| - i)$
different itineraries. This is still an overestimation, though, as it
does not consider the $t_{max}$ constraint, nor does it consider
that once an itinerary has reached $max_{k_j}$ POIs for a
category $k_j$, the set $K_{k_j}$ can be discarded completely
when extending the itinerary further. Introducing a maximum queue length of
$\lmax$, we arrive at the final number 
$\prod_{i=0}^{\lambda-1} \min((|R| - i), \lmax)$.

\section{Experimental Evaluation}
\label{sec:benchmark}

We conducted our experiments on a Linux machine (Arch Linux, kernel version
4.8.4) with an i7-4800MQ CPU running at 2.7 GHz and 24 GB of RAM (of which 15
GB were actually used). The algorithms were written in \Cpp{}, compiled with
\texttt{gcc} 6.2 using \texttt{-O2}. The priority deque was implemented with an
interval heap, the graph representing the city map with an adjacency list. For
comparison with the CLIP algorithm we used black-box testing, configuring it
with the relaxed linear programming solver as suggested in \cite{BHW14}. The
greedy strategy uses Algorithm~\ref{algo:greedyexpand} on an empty itinerary
containing just $s$ and $d$.

We averaged 25 different test runs for every data point in the following plots.
For each test run we selected a start and end node randomly, such that
$\cost{s,d} < \tmax$. When comparing different algorithms, we generated one set
of queries and re-used it for every algorithm. Unless specified otherwise, as
default values for the parameters we two hours for $\tmax$, the largest number
of categories possible for a graph (four or nine), a category constraint of two
for each category, a cut ratio of 1.2, a greedy expansion threshold of 1, an
unbounded queue length, and unlimited run time.

\subsection{Data Sets}

We used several different data sets in our evaluation. The first one was an
artificial one consisting of a grid with 10,000 nodes ($100 \times 100$) with a
total of 3,000 POIs in four different categories. The distance between
neighboring nodes in the grid was 60 seconds. The POIs were randomly scattered
in the grid, had a visiting time between three minutes and one hour and a score
between 1 and 100.

The second artificial network is a spider network having the same number
of nodes and POIs as the grid, but the edges are placed as a 100-sided polygon with 100
levels. The edges between levels are 100 seconds, and the central, and
smallest, polygon has sides of 8 seconds.  The other levels become larger and
larger as in a Euclidean plane.  For a visualization of the grid and spider
networks, see Figure~\ref{fig:spidergrid}.

The two real-world data sets were a map of the
Italian city of Bolzano with a total of 1830 POIs in nine categories and a map
of San Francisco with a total of 1983 POIs in four categories.

\begin{figure}[htb]
{
\centering
\newcommand{\D}{13} 
\newcommand{\U}{5} 
\newdimen\R 
\R=10mm
\newcommand{\A}{360/\D}

\begin{tikzpicture}[x=3.333mm,y=3.333mm]
\draw[step=3.333mm] (0,0) grid (6,6);
\end{tikzpicture}
\hspace{2cm}
\begin{tikzpicture}
[rotate=\A/4]

\foreach \X in {1,...,\D}{
  \draw (\X*\A:\R/\U) -- (\X*\A:\R);
}

\foreach \Y in {0,...,\U}{
  \draw (0:\Y*\R/\U) \foreach \X in {1,...,\D}{
      -- (\X*\A:\Y*\R/\U)
  } -- cycle;
}
\end{tikzpicture}

}

\vspace*{\pullupcaption}
\caption{Grid and spider network example}
\label{fig:spidergrid}
\vspace*{\pullupfigure}
\end{figure}
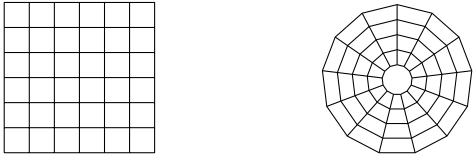

\subsection{Effects of Parameters}

First of all, we illustrate the effects of the different parameters on the
performance of our approximation algorithm. Namely, these are the values for
the greedy expansion threshold, the cut ratio, the queue length limit, and the
run time limit. We also show how the number of categories selected by a user
influences our performance.

\subsubsection{Greedy expansion threshold}

Figure~\ref{fig:greedyexp} shows the impact of the greedy expansion threshold
on the performance of our algorithm (here on a grid network; the results for
the other networks look very similar). As a reminder, the threshold $g$
determines when we stop expanding a partial solution greedily. If the
accumulated cost of an itinerary so far divided by $\tmax$ is less than or
equal to $g$, we do the expansion ($g=0$ meaning we never expand, $g=1$
meaning we always do so).

\begin{figure}[htb]
\begin{tikzpicture}

\begin{axis}[%
graph,
xlabel=Threshold {\smaller[2](percent)},
x coord trafo/.code={\pgfmathparse{\pgfmathresult * 100}},
xtick=data,
ylabel=Score,
legend style={cells={anchor=west}, legend pos=south west},
ymin=0,
]
\pgfplotstableread{experiments/grid_greedy_expansion.txt}\plot
\addplot[color=orange,mark=square] table[x=greedy_cut, y=score] {\plot};
\addlegendentry{Our approach}
\addplot[color=blue] table[x=greedy_cut, y=greedy_score] {\plot};
\addlegendentry{Greedy}
\end{axis}
\end{tikzpicture}
\begin{tikzpicture}
\begin{axis}[%
graphright,
xlabel=Threshold {\smaller[2](percent)},
ylabel=Run time {\smaller[2](seconds)},
legend style={cells={anchor=west}, legend pos=north east},
x coord trafo/.code={\pgfmathparse{\pgfmathresult * 100}},
y coord trafo/.code={\pgfmathparse{\pgfmathresult / 1000}},
xtick=data,
]

\pgfplotstableread{experiments/grid_greedy_expansion.txt}\plot

\addplot[color=orange,mark=square] table[x=greedy_cut, y=runtime] {\plot};
\addlegendentry{Our approach}

\addplot[color=blue] table[x=greedy_cut, y=greedy_runtime] {\plot};
\addlegendentry{Greedy}
\end{axis}
\end{tikzpicture}

\vspace*{\pullupcaption}
\caption{Varying the greedy expansion threshold}
\label{fig:greedyexp}
\vspace*{\pullupfigure}
\end{figure}
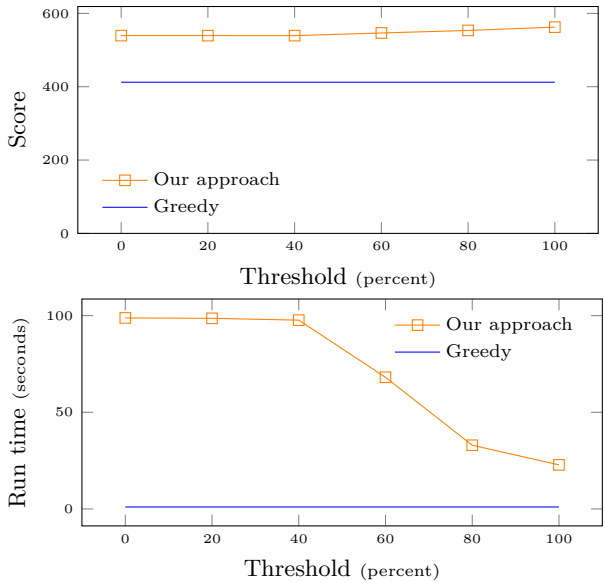

As can be seen in Figure~\ref{fig:greedyexp} (top), the effect on the score is
minimal (for comparison a pure greedy strategy is also shown). However, the
improvements in terms of run time are significant (Figure~\ref{fig:greedyexp},
bottom).  Running a greedy expansion during the execution of the algorithm
updates the best solution found so far much faster, resulting in better
pruning. On average, it was always worth running a greedy expansion right up to
the point when the algorithm finishes, since greedily expanding a partial
solution can be done very quickly. In fact, toward the end of the execution
this can be done even faster, as only one or two more POIs can be added to an
itinerary.

\subsubsection{Cut ratio}
\label{sec:cutratio}

Next, we investigate the impact of the cut ratio.  Increasing the cut ratio
allows us to prune more aggressively, albeit at the price of sometimes losing
the optimal solution. In Figure~\ref{fig:cutratio} we depict the results for
varying the cut ratio from 1 (optimal case) to 2.5 (guaranteeing at least 0.4
of the optimal score).  As expected, the score of the found solutions goes down
when increasing the cut ratio. This is the case for all data sets.
Nevertheless, it does so much slower than the guaranteed lower bound for the
score and it also stays well above the score for the solutions found by the simple
greedy algorithm.
This is especially true for the real-world data sets
(Figures~\ref{fig:cutratio}(c), (d), (e), and (f)). We found that our algorithm
finds good solutions early (more on this in Section~\ref{sec:runtimelimit}
about setting a fixed run time limit).

\begin{figure*}[htb]

\begin{tabularx}{\textwidth}{cXc}

\begin{tikzpicture}

\begin{axis}[%
graph,
xlabel=Cut ratio,
xtick={1.0,1.25,...,3.125},
ylabel=Score,
ymax=620,
legend style={cells={anchor=west}, legend pos=south west},
ymin=0,
]
\pgfplotstableread{experiments/grid_cut_ratio.txt}\plot

\addplot[color=orange,mark=square] table[x=cut_ratio, y=score] {\plot};
\addlegendentry{Our approach}

\addplot[color=blue] table[x=cut_ratio, y=greedy_score] {\plot};
\addlegendentry{Greedy}

\addplot[color=black,mark=|] table[x=cut_ratio, y expr=569.6 / \thisrow{cut_ratio} ] {\plot};
\addlegendentry{Guaranteed lower bound}
\end{axis}
\end{tikzpicture}
&&
\begin{tikzpicture}

\begin{axis}[%
graph,
xlabel=Cut ratio,
xtick={1.0,1.25,...,3.125},
ylabel=Score,
legend style={cells={anchor=west}, legend pos=south west},
ymin=0,
]
\pgfplotstableread{experiments/spider_cut_ratio.txt}\plot

\addplot[color=orange,mark=square] table[x=cut_ratio, y=score] {\plot};
\addlegendentry{Our approach}

\addplot[color=blue] table[x=cut_ratio, y=greedy_score] {\plot};
\addlegendentry{Greedy}

\addplot[color=black,mark=|] table[x=cut_ratio, y expr=321.2 / \thisrow{cut_ratio} ] {\plot};
\addlegendentry{Guaranteed lower bound}
\end{axis}
\end{tikzpicture}
\\
(a) Grid network && (b) Spider network \\
&& \\

\begin{tikzpicture}

\begin{axis}[%
graph,
xlabel=Cut ratio,
xtick={1.0,1.25,...,3.125},
ylabel=Score,
legend style={cells={anchor=west}, legend pos=south west},
ymin=0,
]
\pgfplotstableread{experiments/bozen_cut_ratio.txt}\plot

\addplot[color=orange,mark=square] table[x=cut_ratio, y=score] {\plot};
\addlegendentry{Our approach}

\addplot[color=blue] table[x=cut_ratio, y=greedy_score] {\plot};
\addlegendentry{Greedy}

\addplot[color=black,mark=|] table[x=cut_ratio, y expr=252.32 / \thisrow{cut_ratio} ] {\plot};
\addlegendentry{Guaranteed lower bound}
\end{axis}
\end{tikzpicture}
&&
\begin{tikzpicture}
\raggedright

\begin{axis}[%
graph,
xlabel=Cut ratio,
xtick={1.0,1.25,...,3.125},
ylabel=Score,
legend style={cells={anchor=west}, legend pos=south west},
ymin=0,
]
\pgfplotstableread{experiments/sanfra_cut_ratio.txt}\plot

\addplot[color=orange,mark=square] table[x=cut_ratio, y=score] {\plot};
\addlegendentry{Our approach}

\addplot[color=blue] table[x=cut_ratio, y=greedy_score] {\plot};
\addlegendentry{Greedy}

\addplot[color=black,mark=|] table[x=cut_ratio, y expr=171.92 / \thisrow{cut_ratio} ] {\plot};
\addlegendentry{Guaranteed lower bound}
\end{axis}
\end{tikzpicture}
\\
(c) Bolzano network ($\tmax$ = 2h) && (d) San Francisco network ($\tmax$ = 2h) \\
&& \\

\begin{tikzpicture}

\begin{axis}[%
graph,
xlabel=Cut ratio,
xtick={1.0,1.25,...,3.125},
ylabel=Score,
legend style={cells={anchor=west}, legend pos=south west},
ymin=0,
]
\pgfplotstableread{experiments/bozen_cut_ratio_3h.txt}\plot

\addplot[color=orange,mark=square] table[x=cut_ratio, y=score] {\plot};
\addlegendentry{Our approach}

\addplot[color=blue] table[x=cut_ratio, y=greedy_score] {\plot};
\addlegendentry{Greedy}

\addplot[color=black,mark=|] table[x=cut_ratio, y expr=468.636 / \thisrow{cut_ratio} ] {\plot};
\addlegendentry{Guaranteed lower bound}
\end{axis}
\end{tikzpicture}
&&
\begin{tikzpicture}
\raggedright

\begin{axis}[%
graph,
xlabel=Cut ratio,
xtick={1.0,1.25,...,3.125},
ylabel=Score,
legend style={cells={anchor=west}, legend pos=south west},
ymin=0,
]
\pgfplotstableread{experiments/sanfra_cut_ratio_3h.txt}\plot

\addplot[color=orange,mark=square] table[x=cut_ratio, y=score] {\plot};
\addlegendentry{Our approach}

\addplot[color=blue] table[x=cut_ratio, y=greedy_score] {\plot};
\addlegendentry{Greedy}

\addplot[color=black,mark=|] table[x=cut_ratio, y expr=346.28 / \thisrow{cut_ratio} ] {\plot};
\addlegendentry{Guaranteed lower bound}
\end{axis}
\end{tikzpicture}
\\
(e) Bolzano network ($\tmax$ = 3h) && (f) San Francisco network ($\tmax$ = 3h) \\

\end{tabularx}

\vspace*{\pullupcaption}
\caption{Increasing the cut ratio}
\label{fig:cutratio}
\vspace*{\pullupfigure}
\end{figure*}

The impact of the cut ratio on the run time is much larger, though
(see Figure~\ref{fig:cutratiort}). For some
data sets, namely the grid and San Francisco networks,
increasing the cut ratio brings down the run time
significantly.\footnote{The two hour
  itineraries in the real-world setting for Bolzano and San Francisco are a bit
  too short to notice an effect.} 
Grid networks, including the artificial grid and San Francisco,
whose map is composed of many grid-like structures, have a very long run time
for small cut ratios, while spider networks do not exhibit this behavior.
(Bolzano, which has grown organically over centuries around an old city center,
shares some of the features of a spider network.) 
The different edge lengths make the central part of the spider
network much easier to explore, whereas the external parts are more
isolated. As a consequence, the computation of \texttt{extra}($\cal I$) is
more precise, meaning that we get a much more accurate picture about the
potential scores of the yet unexpanded solutions that are still in the queue,
making it easier to prune partial solutions. In a grid-like network, there are
lots of alternatives, making it more difficult to compute \texttt{extra}($\cal I$)
accurately, so the potential scores of some unexpanded solutions may still be
high, and we need more time to check them.

\begin{figure*}[htb]
\centering

\begin{tikzpicture}
\begin{axis}[%
graphright,
xlabel=Cut ratio,
xtick={1.0,1.25,...,3.125},
width=\textwidth - 15em,
ylabel=Run time {\smaller[2](second)},
scaled y ticks=manual:{}{\pgfmathparse{#1 / 1000}},
legend style={ cells={anchor=west}, at={(1,1)}, anchor=north west},
ymax=300000,
]

\pgfplotstableread{experiments/grid_cut_ratio.txt}\plotgrid
\pgfplotstableread{experiments/spider_cut_ratio.txt}\plotspider
\pgfplotstableread{experiments/bozen_cut_ratio_3h.txt}\plotbozenthreeh
\pgfplotstableread{experiments/bozen_cut_ratio.txt}\plotbozen
\pgfplotstableread{experiments/sanfra_cut_ratio_3h.txt}\plotsanfrathreeh
\pgfplotstableread{experiments/sanfra_cut_ratio.txt}\plotsanfra

\addplot[color=blue] table[x=cut_ratio, y=greedy_runtime] {\plotsanfrathreeh};
\addlegendentry{Greedy}

\addplot[color=orange,mark=triangle] table[x=cut_ratio, y=runtime] {\plotgrid};
\addlegendentry{Grid}

\addplot[color=green,mark=square] table[x=cut_ratio, y=runtime] {\plotspider};
\addlegendentry{Spider}

\addplot[densely dashed,color=blue,mark=diamond] table[x=cut_ratio, y=runtime] {\plotbozen};
\addlegendentry{Bolzano (2h)}

\addplot[color=blue,mark=diamond] table[x=cut_ratio, y=runtime] {\plotbozenthreeh};
\addlegendentry{Bolzano (3h)}

\addplot[densely dashed, color=red,mark=o] table[x=cut_ratio, y=runtime] {\plotsanfra};
\addlegendentry{San Francisco (2h)}

\addplot[color=red,mark=o] table[x=cut_ratio, y=runtime] {\plotsanfrathreeh};
\addlegendentry{San Francisco (3h)}
\end{axis}

\end{tikzpicture}

\vspace*{\pullupcaption}
\caption{Effect of the cut ratio on run time}
\label{fig:cutratiort}
\vspace*{\pullupfigure}
\end{figure*}
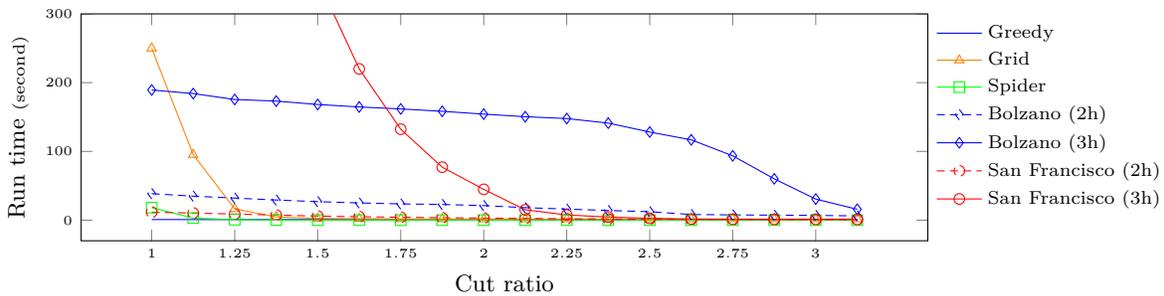

\subsubsection{Queue length}

Now we look at the impact of bounding the queue length to shorten the run time
of our algorithm. The left-hand column of Figure~\ref{fig:dequelength}, labeled
(a), shows the results for the grid network, the right-hand column, labeled
(b), those for the spider network. The cut ratio is set to 1 to measure only
the impact of the queue length limit.\footnote{Due to the cut ratio, we did not
run this experiment for the real-world data sets, as is would have taken too
much time.} The size of the queue on the x-axis is measured relatively to the
number of POIs in the graph and the horizontal line at the top of the figure is
the optimal result when running our algorithm with an unbounded queue. As can
be seen in the top row of Figure~\ref{fig:dequelength}, having a very short
queue lowers the score slightly.  We do not start the plot y-axis at 0 to make
the small differences of the score more visible.  That is also the reason why
we did not add the score of the greedy algorithm (it is 412 for
Figure~\ref{fig:dequelength}(a) and 252 for Figure~\ref{fig:dequelength}(b)).
Unsurprisingly, the run time goes up with increasing queue length (see bottom
row of Figure~\ref{fig:dequelength}) and it will eventually reach the run time
of the unbounded queue.

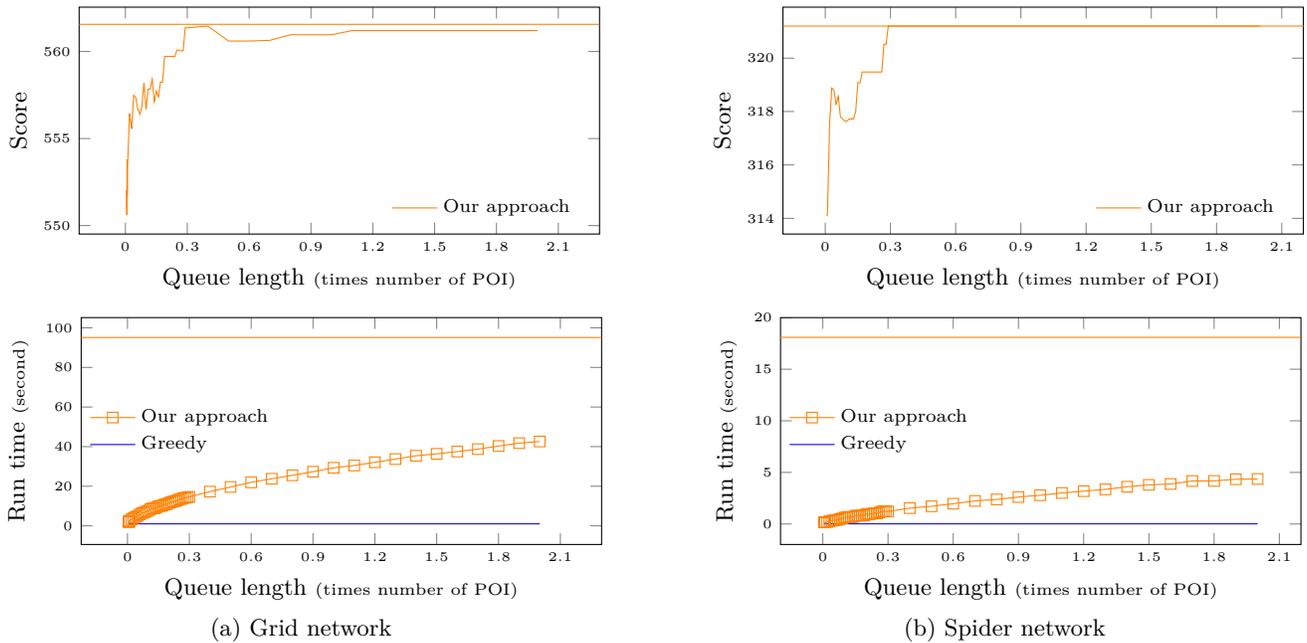
\begin{figure*}[htb]

\begin{tabularx}{\textwidth}{cXc}

\begin{tikzpicture}
\begin{axis}[%
graph,
xlabel=Queue length {\smaller[2](times number of POI)},
xtick={0.0,0.3,...,2.2},
xmax=2.3,
ylabel=Score,
legend style={cells={anchor=west}, legend pos=south east},
]
\pgfplotstableread{experiments/grid_queue_length_new2.txt}\plot

\draw[color=orange,thin] (axis cs:\pgfkeysvalueof{/pgfplots/xmin},561.56) -- (axis cs:\pgfkeysvalueof{/pgfplots/xmax},561.56);

\addplot[color=orange] table[x=queue_length, y=score] {\plot};
\addlegendentry{Our approach}

\end{axis}
\end{tikzpicture}
&&
\begin{tikzpicture}
\begin{axis}[%
graph,
xlabel=Queue length {\smaller[2](times number of POI)},
xtick={0.0,0.3,...,2.2},
ylabel=Score,
legend style={cells={anchor=west}, legend pos=south east},
]
\pgfplotstableread{experiments/spider_queue_length_f.txt}\plot

\draw[color=orange,thin] (axis cs:\pgfkeysvalueof{/pgfplots/xmin},321.2) -- (axis cs:\pgfkeysvalueof{/pgfplots/xmax},321.2);

\addplot[color=orange] table[x=queue_length, y=score] {\plot};
\addlegendentry{Our approach}
\end{axis}
\end{tikzpicture}
\\
\begin{tikzpicture}
\begin{axis}[%
graphright,
xlabel=Queue length {\smaller[2](times number of POI)},
xtick={0.0,0.3,...,2.2},
xmax=2.3,
ymax=105152,
ylabel=Run time {\smaller[2](second)},
scaled y ticks=manual:{}{\pgfmathparse{#1 / 1000}},
ytick={0,20000,40000,60000,80000,100000},
legend style={ cells={anchor=west}, at={(0,0.5)}, anchor=west},
]
\pgfplotstableread{experiments/grid_queue_length_new2.txt}\plot

\draw[color=orange,thin] (axis cs:\pgfkeysvalueof{/pgfplots/xmin},95116) -- (axis cs:\pgfkeysvalueof{/pgfplots/xmax},95116);

\addplot[color=orange,mark=square] table[x=queue_length, y=runtime] {\plot};
\addlegendentry{Our approach}

\addplot[color=blue] table[x=queue_length, y=greedy_runtime] {\plot};
\addlegendentry{Greedy}
\end{axis}
\end{tikzpicture}
&&
\begin{tikzpicture}
\begin{axis}[%
graphright,
xlabel=Queue length {\smaller[2](times number of POI)},
xtick={0.0,0.3,...,2.2},
ymax=20000,
ylabel=Run time {\smaller[2](second)},
scaled y ticks=manual:{}{\pgfmathparse{#1 / 1000}},
legend style={ cells={anchor=west}, at={(0,0.5)}, anchor=west},
]
\pgfplotstableread{experiments/spider_queue_length_f.txt}\plot

\draw[color=orange,thin] (axis cs:\pgfkeysvalueof{/pgfplots/xmin},18091) -- (axis cs:\pgfkeysvalueof{/pgfplots/xmax},18091);

\addplot[color=orange,mark=square] table[x=queue_length, y=runtime] {\plot};
\addlegendentry{Our approach}

\addplot[color=blue] table[x=queue_length, y=greedy_runtime] {\plot};
\addlegendentry{Greedy}
\end{axis}
\end{tikzpicture}
\\
(a) Grid network && (b) Spider network \\
\end{tabularx}

\caption{Limiting the queue length}

\vspace*{\pullupcaption}
\label{fig:dequelength}
\vspace*{\pullupfigure}
\end{figure*}

\subsubsection{Run Time Limit}
\label{sec:runtimelimit}

Bounding the queue length is an indirect way of controlling the run time, in
Figure~\ref{fig:limitruntimegrid} we show results for setting an explicit run
time limit. The orange line at the top of of each plot represents the optimal
score. We can see that even for very short run time limits we get results
for every data set that are very close to the optimal score. What this means is
that our algorithm finds very good solutions early in the search process and
then, if we let it continue, spends the remaining time basically verifying that
there are not significantly better solutions.  We also clearly outperform the
greedy heuristic in terms of score (the greedy algorithm takes between 0.2 and
1.5 seconds to find a solution).

When limiting the run time, we cannot guarantee the quality of the attained
score a priori. Nevertheless, when the algorithm stops running, toghether with
the score it returns the factor $\alpha$, that tell us which proportion of the
optimal score we have reached at least (for details, see
Section~\ref{sec:lowerboundscore}).  We have plotted the lower bound $\alpha$
in Figure~\ref{fig:limitruntimegrid}, please note that for this measure the
scale on the right-hand side of the y-axis is used. This has important
implications for running our algorithm: we basically have an anytime algorithm.
We can stop it at any time and get a solution and its quality.  Eventually we
continue with the execution and investigate the current state a short time
later, repeating the process.  Consequently, we are able to balance the quality
of the solution and the run time on the fly.

Looking at Figure~\ref{fig:limitruntimegrid}, we see that the results for two
of the data sets, the spider network in Figure~\ref{fig:limitruntimegrid}(b)
and the Bolzano network in Figure~\ref{fig:limitruntimegrid}(c) are somewhat
different, which comes as a surprise, as we mentioned that the Bolzano network
shares some of the properties of a spider network in
Section~\ref{sec:cutratio}.  However, there are other effects at work here,
too.  For the Bolzano network the larger number of categories is responsible
for the lower guarantee. This makes it more difficult to compute
\texttt{extra}($\cal I$) accurately, making it more difficult to discard
potential unexpanded solutions.  (We go into more details about the number of
categories in the following section.) We would like to point out that the found
solutions are close to the optimal for every data set, only the lower bounds
for the score are affected differently.

\begin{figure*}[htb]

\begin{tabularx}{\textwidth}{cXc}

\begin{tikzpicture}

\begin{axis}[%
graph,
width=75mm,
xlabel=Run time limit {\smaller[2](seconds)},
ylabel=Score,
axis y line*=none,
ymin=0,
ymax=603.58,
]
\pgfplotstableread{experiments/grid_runtime_limit_new.txt}\plot

\draw[color=orange,thin] (axis cs:\pgfkeysvalueof{/pgfplots/xmin},586.0) -- (axis cs:\pgfkeysvalueof{/pgfplots/xmax},586.0);

\addplot[color=orange,mark=square] table[x=runtime_limit, y=score] {\plot}; \label{ourapproach}
\addplot[color=blue] table[x=runtime_limit, y=greedy_score] {\plot}; \label{greedy}
\end{axis}

\begin{axis}[%
graph,
width=75mm,
ylabel=Quality,
ymin=.0,
ymax=1.03,
axis x line=none,
axis y line*=right,
legend style={at={(0.02,0.3)},anchor=west,cells={anchor=west} },
]
\pgfplotstableread{experiments/grid_runtime_limit_new.txt}\plot

\addlegendimage{/pgfplots/refstyle=greedy}\addlegendentry{Greedy}
\addlegendimage{/pgfplots/refstyle=ourapproach}\addlegendentry{Our approach}

\addplot[color=olive,mark=o] table[x=runtime_limit, y expr=1. / \thisrow{quality}] {\plot};
\addlegendentry{Solution quality $\alpha$}

\end{axis}
\end{tikzpicture}
&&
\begin{tikzpicture}

\begin{axis}[%
graph,
width=75mm,
xlabel=Run time limit {\smaller[2](seconds)},
ylabel=Score,
axis y line*=none,
ymin=0,
ymax=330.836,
]
\pgfplotstableread{experiments/spider_runtime_limit.txt}\plot

\draw[color=orange,thin] (axis cs:\pgfkeysvalueof{/pgfplots/xmin},321.2) -- (axis cs:\pgfkeysvalueof{/pgfplots/xmax},321.2);

\addplot[color=orange,mark=square] table[x=runtime_limit, y=score] {\plot}; \label{ourapproach}
\addplot[color=blue] table[x=runtime_limit, y=greedy_score] {\plot}; \label{greedy}
\end{axis}

\begin{axis}[%
graph,
width=75mm,
ylabel=Quality,
ymin=.0,
ymax=1.03,
axis x line=none,
axis y line*=right,
legend style={at={(0.02,0.3)},anchor=west,cells={anchor=west} },
]
\pgfplotstableread{experiments/spider_runtime_limit.txt}\plot

\addlegendimage{/pgfplots/refstyle=greedy}\addlegendentry{Greedy}
\addlegendimage{/pgfplots/refstyle=ourapproach}\addlegendentry{Our approach}

\addplot[color=olive,mark=o] table[x=runtime_limit, y expr=1. / \thisrow{quality}] {\plot};
\addlegendentry{Solution quality $\alpha$}

\end{axis}
\end{tikzpicture}
\\
(a) Grid network && (b) Spider network \\
&& \\
\begin{tikzpicture}

\begin{axis}[%
graph,
width=75mm,
xlabel=Run time limit {\smaller[2](seconds)},
ylabel=Score,
axis y line*=none,
ymin=0,
ymax=262.980,
]
\pgfplotstableread{experiments/bozen_runtime_limit.txt}\plot

\draw[color=orange,thin] (axis cs:\pgfkeysvalueof{/pgfplots/xmin},255.32) -- (axis cs:\pgfkeysvalueof{/pgfplots/xmax},255.32);

\addplot[color=orange,mark=square] table[x=runtime_limit, y=score] {\plot}; \label{ourapproach}
\addplot[color=blue] table[x=runtime_limit, y=greedy_score] {\plot}; \label{greedy}
\end{axis}

\begin{axis}[%
graph,
width=75mm,
ylabel=Quality,
ymin=.0,
ymax=1.03,
axis x line=none,
axis y line*=right,
legend style={cells={anchor=west}, legend pos=south east},
]
\pgfplotstableread{experiments/bozen_runtime_limit.txt}\plot

\addlegendimage{/pgfplots/refstyle=greedy}\addlegendentry{Greedy}
\addlegendimage{/pgfplots/refstyle=ourapproach}\addlegendentry{Our approach}

\addplot[color=olive,mark=o] table[x=runtime_limit, y expr=1. / \thisrow{quality}] {\plot};
\addlegendentry{Solution quality $\alpha$}

\end{axis}
\end{tikzpicture}
&& 
\begin{tikzpicture}

\begin{axis}[%
graph,
width=75mm,
xlabel=Run time limit {\smaller[2](seconds)},
ylabel=Score,
axis y line*=none,
ymin=0,
ymax=177.078,
]
\pgfplotstableread{experiments/sanfra_runtime_limit.txt}\plot

\draw[color=orange,thin] (axis cs:\pgfkeysvalueof{/pgfplots/xmin},171.92) -- (axis cs:\pgfkeysvalueof{/pgfplots/xmax},171.92);

\addplot[color=orange,mark=square] table[x=runtime_limit, y=score] {\plot}; \label{ourapproach}
\addplot[color=blue] table[x=runtime_limit, y=greedy_score] {\plot}; \label{greedy}
\end{axis}

\begin{axis}[%
graph,
width=75mm,
ylabel=Quality,
ymin=.0,
ymax=1.03,
axis x line=none,
axis y line*=right,
legend style={cells={anchor=west}, legend pos=south east},
]
\pgfplotstableread{experiments/sanfra_runtime_limit.txt}\plot

\addlegendimage{/pgfplots/refstyle=greedy}\addlegendentry{Greedy}
\addlegendimage{/pgfplots/refstyle=ourapproach}\addlegendentry{Our approach}

\addplot[color=olive,mark=o] table[x=runtime_limit, y expr=1. / \thisrow{quality}] {\plot};
\addlegendentry{Solution quality $\alpha$}

\end{axis}
\end{tikzpicture}
\\
(c) Bolzano network (2h) && (d) San Francisco network (2h) \\
\end{tabularx}

\vspace*{\pullupcaption}
\caption{Limiting the run time}
\label{fig:limitruntimegrid}
\vspace*{\pullupfigure}
\end{figure*}
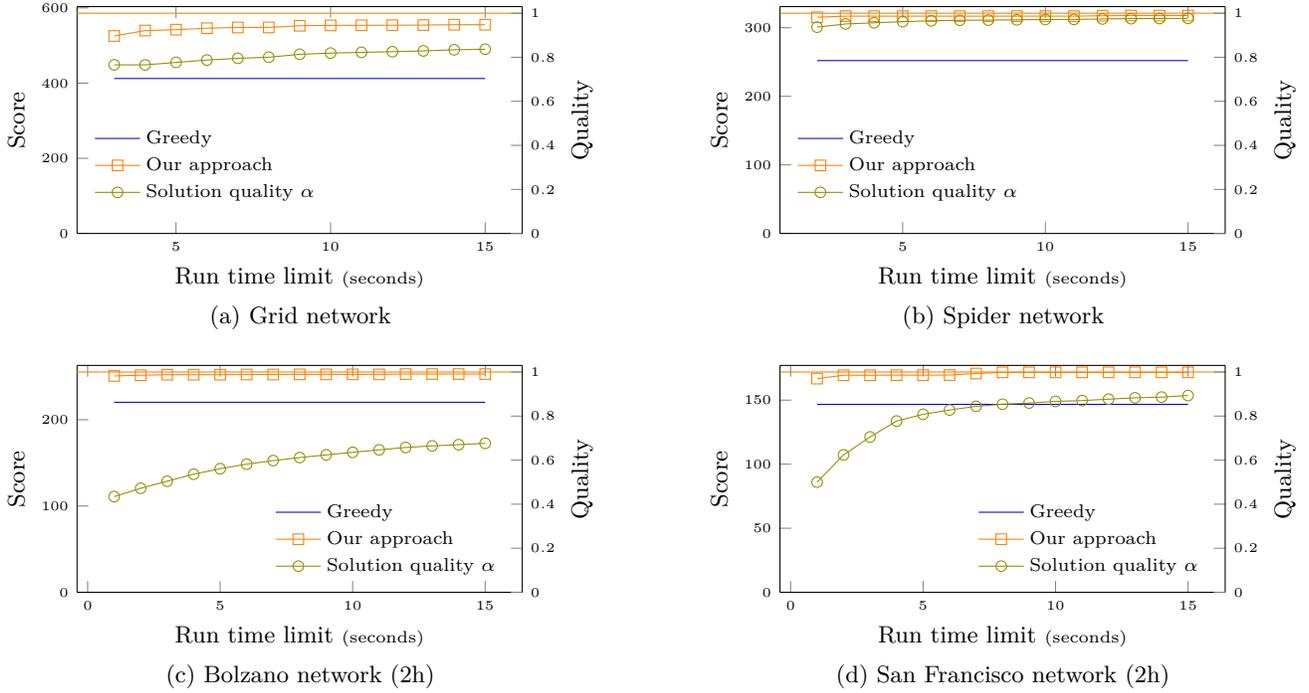

\subsubsection{Impact of Categories}

For the final parameter, we illustrate the impact of categories on the run
time. While our approximation algorithms scale well in terms of the itinerary
length ($\tmax$), as we will see in Section~\ref{sec:realworld}, this is not
independent of the number of categories. Figure~\ref{fig:nrcategoriesscore}(a)
shows the effects of increasing the number of categories. By putting a limit on
the run time, we can keep the execution time of the algorithm low, while still
generating very good solutions. In the upper diagram, we only show the plots
for two and eight categories to keep it uncluttered. In the lower diagram, the
run time for two categories actually goes down for longer itineraries. At the
end of visiting two POIs of two categories each, we actually start having more
and more spare time, which makes the problem easier to solve.

Figure~\ref{fig:nrcategoriesscore}(b) shows the effect of increasing the maximum
constraint of categories. Again, we quickly generate good solutions in a
scalable way, due to the explicit limit on the run time.

\begin{figure*}[htb]

\begin{tabularx}{\textwidth}{cXc}
\begin{tikzpicture}
\begin{axis}[%
graphright,
xlabel=\tmax{},
ylabel=Score,
legend style={cells={anchor=west}, legend pos=north west},
]
\pgfplotstableread{experiments/categories/bozen_2cats_5sec.txt}\plottwo
\pgfplotstableread{experiments/categories/bozen_4cats_5sec.txt}\plotfour
\pgfplotstableread{experiments/categories/bozen_6cats_5sec.txt}\plotsix
\pgfplotstableread{experiments/categories/bozen_8cats_5sec.txt}\ploteight

\addplot[densely dashed, color=darkpink,mark=pentagon] table[x=t_max, y=greedy_score] {\ploteight};
\addlegendentry{8 categories (greedy)}

\addplot[color=darkpink,mark=pentagon] table[x=t_max, y=score] {\ploteight};
\addlegendentry{8 categories}

\addplot[densely dashed,color=orange,mark=triangle] table[x=t_max, y=greedy_score] {\plottwo};
\addlegendentry{2 categories (greedy)}

\addplot[color=orange,mark=triangle] table[x=t_max, y=score] {\plottwo};
\addlegendentry{2 categories}

\end{axis}
\end{tikzpicture}
&&
\begin{tikzpicture}
\begin{axis}[%
graph,
xlabel=$\max_{k}$,
xtick={2,4,...,8},
ymin=-10,
ylabel=Score,
legend style={ cells={anchor=west}, legend pos=south west },
]
\pgfplotstableread{experiments/bozen_maxk.txt}\plot

\addplot[color=orange,mark=triangle] table[x=maxk, y=score] {\plot};
\addlegendentry{Our approach}

\addplot[color=blue] table[x=maxk, y=greedy_score] {\plot};
\addlegendentry{Greedy}
\end{axis}
\end{tikzpicture} 
\\
\begin{tikzpicture}
\begin{axis}[%
graphright,
xlabel=\tmax{},
ylabel=Run time {\smaller[2](seconds)},
ymax=6,
legend style={cells={anchor=west}, legend pos=north west},
]
\pgfplotstableread{experiments/categories/bozen_2cats_5sec.txt}\plottwo
\pgfplotstableread{experiments/categories/bozen_4cats_5sec.txt}\plotfour
\pgfplotstableread{experiments/categories/bozen_6cats_5sec.txt}\plotsix
\pgfplotstableread{experiments/categories/bozen_8cats_5sec.txt}\ploteight

\addplot[color=blue] table[x=t_max, y expr=\thisrow{greedy_runtime} / 1000] {\ploteight};
\addlegendentry{Greedy}

\addplot[color=darkpink,mark=pentagon] table[x=t_max, y expr=\thisrow{runtime} / 1000] {\ploteight};
\addlegendentry{8 categories}

\addplot[color=olive,mark=diamond] table[x=t_max, y expr=\thisrow{runtime} / 1000] {\plotsix};
\addlegendentry{6 categories}

\addplot[color=bronze,mark=square] table[x=t_max, y expr=\thisrow{runtime} / 1000] {\plotfour};
\addlegendentry{4 categories}

\addplot[color=orange,mark=triangle] table[x=t_max, y expr=\thisrow{runtime} / 1000] {\plottwo};
\addlegendentry{2 categories}
\end{axis}
\end{tikzpicture} 
&&
\begin{tikzpicture}
\begin{axis}[%
graph,
xlabel=$\max_{k}$,
xtick={2,4,...,8},
ylabel=Run time {\smaller[2](seconds)},
ymax=6,
legend style={cells={anchor=west}, legend pos=north west},
]
\pgfplotstableread{experiments/bozen_maxk.txt}\plot

\addplot[color=orange,mark=triangle] table[x=maxk, y expr=\thisrow{runtime} / 1000] {\plot};
\addlegendentry{Our approach}

\addplot[color=blue] table[x=maxk, y expr= \thisrow{greedy_runtime} / 1000] {\plot};
\addlegendentry{Greedy}
\end{axis}
\end{tikzpicture} \\
(a) Varying number of categories && (b) Varying $\max_k$ constraints \\
\end{tabularx}

\vspace*{\pullupcaption}
\caption{Impact of Categories}
\label{fig:nrcategoriesscore}
\vspace*{\pullupfigure}
\end{figure*}
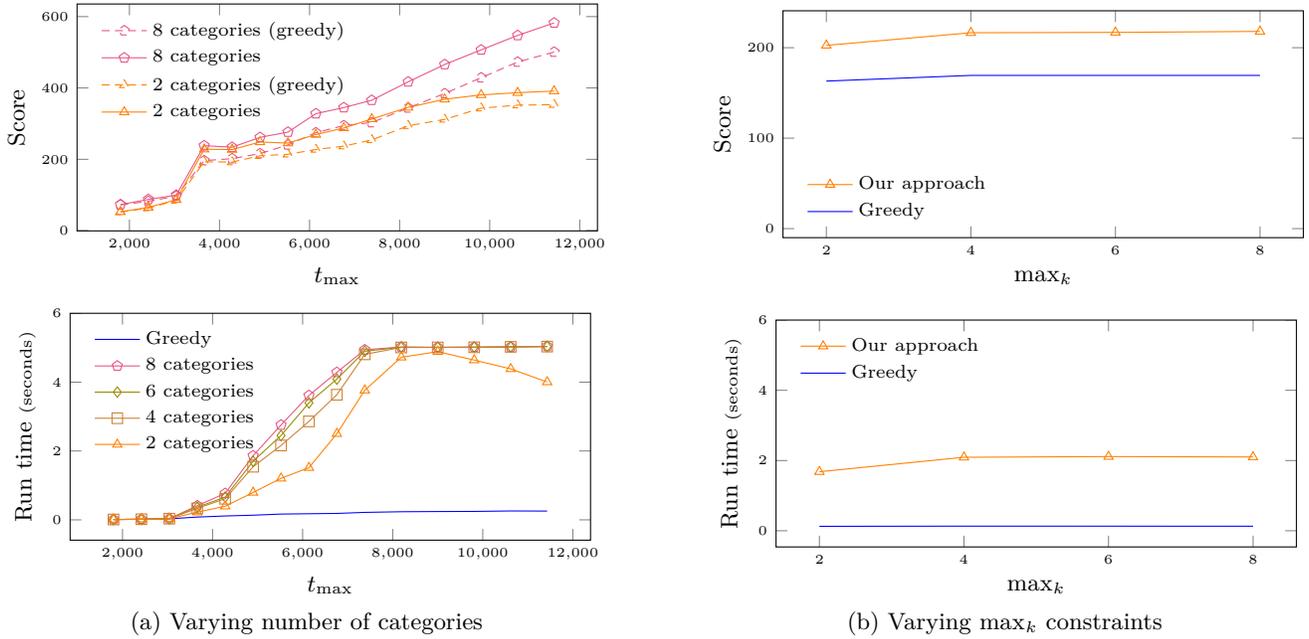

\subsection{Comparison with Competitors}
\label{sec:realworld}

Here we compare our algorithm to the state-of-the-art. We distinguish two
different cases: (a) the artificial networks (grid and spider) and (b) the
real-world data sets (Bolzano and San Francisco). For the real-world data sets
we also run more realistic queries with a larger time constraint $\tmax$ and
for that reason do not compare it to the optimal algorithm (as its run time
explodes for long itineraries).

\subsubsection{Grid and Spider Networks}

Figure~\ref{fig:againstexactgrid} illustrates the results of running our
algorithm with a cut factor of 1.2 against the optimal algorithm (the
left-hand column, labeled (a), shows the results for the grid network, the
right-hand column, labeled (b), shows those for the spider network). In terms
of the score our approach outperforms the greedy approach and comes close to
the optimal solution, while at the same time having a much better run-time
performance than the optimal algorithm. The final measurements for large
values of $\tmax$ are missing, since we aborted runs taking longer than ten
minutes to complete.

\begin{figure*}[htb]

\begin{tabularx}{\textwidth}{cXc}

\begin{tikzpicture}

\begin{axis}[%
graph,
xlabel=\tmax{} {\smaller[2](seconds)},
ylabel=Score,
ymin=-25,
legend style={cells={anchor=west}, legend pos=north west},
]
\pgfplotstableread{experiments/grid_runtime_unlimited.txt}\plot
\pgfplotstableread{experiments/grid_exact_results.txt}\plotexact

\addplot[color=green,mark=o] table[x=t_max, y=score] {\plotexact};
\addlegendentry{Exact algorithm}

\addplot[color=orange,mark=square] table[x=t_max, y=score] {\plot};
\addlegendentry{Our approach}

\addplot[color=blue] table[x=t_max, y=greedy_score] {\plot};
\addlegendentry{Greedy}

\end{axis}
\end{tikzpicture}
&&
\begin{tikzpicture}

\begin{axis}[%
graph,
xlabel=\tmax{} {\smaller[2](seconds)},
ylabel=Score,
ymin=-25,
legend style={cells={anchor=west}, legend pos=north west},
]
\pgfplotstableread{experiments/spider_us.txt}\plot
\pgfplotstableread{experiments/spider_opt.txt}\plotexact

\addplot[color=green,mark=o] table[x=t_max, y=score] {\plotexact};
\addlegendentry{Exact algorithm}

\addplot[color=orange,mark=square] table[x=t_max, y=score] {\plot};
\addlegendentry{Our approach}

\addplot[color=blue] table[x=t_max, y=greedy_score] {\plot};
\addlegendentry{Greedy}

\end{axis}
\end{tikzpicture} \\

\begin{tikzpicture}
\begin{axis}[%
graphright,
xlabel=\tmax{} {\smaller[2](seconds)},
ylabel=Run time {\smaller[2](seconds)},
scaled y ticks=manual:{}{\pgfmathparse{#1 / 1000}},
ytick={0,25000,50000,75000,100000},
ymin=-5000,
legend style={cells={anchor=west}, legend pos=north west},
]
\pgfplotstableread{experiments/grid_runtime_unlimited.txt}\plot
\pgfplotstableread{experiments/grid_exact_results.txt}\plotexact

\addplot[color=green,mark=o] table[x=t_max, y=optimal_algorithm] {\plotexact};
\addlegendentry{Exact algorithm}

\addplot[color=orange,mark=square] table[x=t_max, y=runtime] {\plot};
\addlegendentry{Our approach}

\addplot[color=blue] table[x=t_max, y=greedy_runtime] {\plot};
\addlegendentry{Greedy}
\end{axis}
\end{tikzpicture}
&&
\begin{tikzpicture}
\begin{axis}[%
graphright,
xlabel=\tmax{} {\smaller[2](seconds)},
ylabel=Run time {\smaller[2](seconds)},
scaled y ticks=manual:{}{\pgfmathparse{#1 / 1000}},
legend style={cells={anchor=west}, legend pos=north west},
]
\pgfplotstableread{experiments/spider_us.txt}\plot
\pgfplotstableread{experiments/spider_opt.txt}\plotexact

\addplot[color=green,mark=o] table[x=t_max, y=runtime] {\plotexact};
\addlegendentry{Exact algorithm}

\addplot[color=orange,mark=square] table[x=t_max, y=runtime] {\plot};
\addlegendentry{Our approach}

\addplot[color=blue] table[x=t_max, y=greedy_runtime] {\plot};
\addlegendentry{Greedy}

\end{axis}
\end{tikzpicture}
\\
(a) Grid network && (b) Spider network \\
\end{tabularx}

\vspace*{\pullupcaption}
\caption{Comparison with the exact algorithm}
\label{fig:againstexactgrid}
\vspace*{\pullupfigure}
\end{figure*}

\subsubsection{Real-world Data Sets}

We now move to the real-world data sets. Figure~\ref{fig:bozenmap}(a), found
in the left-hand column, compares three different variants of our algorithm
(cut factor, bounded queue, and run time limit) with the state-of-the-art
algorithm for orienteering with categories, 
CLIP on a map of Bolzano. In terms of the score (upper part of
Figure~\ref{fig:bozenmap}(a)) the three variants of our algorithm are almost
identical: merely one measurement resulted in a difference two digits after the
decimal point (thus we only depict one curve). Our algorithm outperforms both,
CLIP and the greedy heuristic, the latter by a large margin. When we look at
the run time (lower part of Figure~\ref{fig:bozenmap}(a)), we see huge
differences, though. The only competitive algorithms are the ones limiting the
run time in some form, either directly or by limiting the queue length.

\begin{figure*}[htb]

\begin{tabularx}{\textwidth}{cXc}

\begin{tikzpicture}

\begin{axis}[%
graph,
xlabel=\tmax{} {\smaller[2](seconds)},
ylabel=Score,
legend style={cells={anchor=west}, legend pos=north west},
ymin=-25,
ymax=470,
]
\pgfplotstableread{experiments/bozen_tmax_queue_1.txt}\plotbozen
\pgfplotstableread{experiments/clip_bozen_tmax.txt}\plotclipbozen

\addplot[color=orange,mark=square] table[x=t_max, y=score] {\plotbozen};
\addlegendentry{Our approach, all three settings}

\addplot[color=blue] table[x=t_max, y=greedy_score] {\plotbozen};
\addlegendentry{Greedy}

\addplot[color=magenta,mark=triangle] table[x=tmax, y=score] {\plotclipbozen};
\addlegendentry{Clip}

\end{axis}
\end{tikzpicture}
&&
\begin{tikzpicture}

\begin{axis}[%
graph,
xlabel=\tmax{} {\smaller[2](seconds)},
ylabel=Score,
ymin=-30,
legend style={cells={anchor=west}, legend pos=south west},
]
\pgfplotstableread{experiments/sfrans_tmax.txt}\plot
\pgfplotstableread{experiments/clip_sfrans_tmax.txt}\plotclip

\addplot[color=magenta,mark=triangle] table[x=t_max, y=score] {\plotclip};
\addlegendentry{Clip}

\addplot[color=orange,mark=square] table[x=t_max, y=score] {\plot};
\addlegendentry{Our approach}

\addplot[color=blue] table[x=t_max, y=greedy_score] {\plot};
\addlegendentry{Greedy}
\end{axis}
\end{tikzpicture}
\\
\begin{tikzpicture}
\begin{axis}[%
graphright,
xlabel=\tmax{} {\smaller[2](seconds)},
ylabel=Run time {\smaller[2](seconds)},
scaled y ticks=manual:{}{\pgfmathparse{#1 / 1000}},
ymin=-10000,
ymax=200000,
legend style={cells={anchor=west}, legend pos=north west},
]
\pgfplotstableread{experiments/bozen_tmax_queue_limited_runtime.txt}\plotbozenlruntime
\pgfplotstableread{experiments/bozen_tmax_queue_unbound.txt}\plotbozenunbound
\pgfplotstableread{experiments/bozen_tmax_queue_1.txt}\plotbozenone
\pgfplotstableread{experiments/clip_bozen_tmax.txt}\plotclipbozen

\addplot[color=darkpink,mark=diamond] table[x=t_max, y=runtime] {\plotbozenlruntime};
\addlegendentry{Our approach, limited run time}

\addplot[color=olive,mark=pentagon] table[x=t_max, y=runtime] {\plotbozenone};
\addlegendentry{Our approach, bounded deque}

\addplot[color=orange,mark=square] table[x=t_max, y=runtime] {\plotbozenunbound};
\addlegendentry{Our approach, cut ratio}

\addplot[color=magenta,mark=triangle] table[x=tmax, y=runtime] {\plotclipbozen};
\addlegendentry{Clip}

\addplot[color=blue] table[x=t_max, y=greedy_runtime] {\plotbozenone};
\addlegendentry{Greedy}
\end{axis}

\end{tikzpicture}
&&
\begin{tikzpicture}
\begin{axis}[%
graphright,
xlabel=\tmax{} {\smaller[2](seconds)},
ylabel=Run time {\smaller[2](seconds)},
scaled y ticks=manual:{}{\pgfmathparse{#1 / 1000}},
ymin=-5000,
legend style={cells={anchor=west}, legend pos=north west},
]
\pgfplotstableread{experiments/sfrans_tmax.txt}\plot
\pgfplotstableread{experiments/clip_sfrans_tmax.txt}\plotclip

\addplot[color=magenta,mark=triangle] table[x=t_max, y=runtime] {\plotclip};
\addlegendentry{Clip}

\addplot[color=orange,mark=square] table[x=t_max, y=runtime] {\plot};
\addlegendentry{Our approach}

\addplot[color=blue] table[x=t_max, y=greedy_runtime] {\plot};
\addlegendentry{Greedy}
\end{axis}
\end{tikzpicture}
\\
(a) Bolzano network && (b) San Francisco network \\
\end{tabularx}

\vspace*{\pullupcaption}
\caption{Comparison with CLIP}
\label{fig:bozenmap}
\vspace*{\pullupfigure}
\end{figure*}


In the final experiment we show the full strength of our approach for running
queries with a large time constraint $\tmax$. Here we run a blended version of
our algorithm, combining a run time limit of five seconds with a queue length
of half the number of total POIs in the graph. At first glance,
combining a run time limit with a fixed queue length seems redundant. However,
from a certain run time limit onwards, the queue length may become
quite large and we may want to restrict it.

Figure~\ref{fig:bozenmap}(b) shows results for the San Francisco data set,
comparing CLIP to our blended variant. In terms of the score, our approach
outperforms both, CLIP and the greedy heuristic. However, we get a better score
than CLIP with a shorter run time, demonstrating that our technique scales much
better than CLIP. In fact, as we were running the experiments on weaker
hardware compared to the findings in \cite{BHW14}, we could not replicate the
ten-hour itineraries for CLIP shown in \cite{BHW14} in reasonable time. This
clearly shows that our algorithm is much more suitable for deployment on
mobile devices, which have limited computational resources. We also illustrate
that limiting the queue length (to a certain extent) does not have adverse
effects on the run-time-limited variant.

\section{Conclusion and Future Work}
\label{sec:concl}

We have developed an effective and efficient approximation algorithm for
solving the orienteering problem with category constraints (OPMPC) by
applying a best-first strategy, blending it with greedy search, and then
limiting its run time. One major advantage of our technique over CLIP, the
state-of-the-art approach for OPMPC, is the fact that our technique is an
anytime algorithm, which
immediately starts to generate solutions with quality guarantees, as it keeps
track of potential scores. Consequently, we can run our algorithm for a fixed
time or until a certain quality level has been reached, whichever comes first.

Nevertheless, we still see some room for improvement. Profiling our algorithm,
we noticed that we spend a considerable amount of time (about 40\%)
calculating the potential score. If we were able to do this more efficiently,
maybe by parallelizing the task, we would be able to create an even more
efficient algorithm. For longer itineraries it may also be interesting to move
to more efficient techniques for computing itineraries from sets of POIs. 
On a more general level, our approach
could also be viable for other orienteering variants, such as the team
orienteering problem or orienteering with time windows.

\section*{Acknowledgments} The icons in Figure~\ref{fig:ex2} are from
\url{http://icons8.com} and used under the
Creative Commons Attribution-NoDerivs 3.0 Unported (CC BY-ND 3.0) License. To
view a copy of the license, visit 
\url{https://creativecommons.org/licenses/by-nd/3.0/}.

\bibliographystyle{abbrv}
\bibliography{iip}

\begin{thebibliography}{10}

\bibitem{BCKLMM07}
A.~Blum, S.~Chawla, D.~R. Karger, T.~Lane, A.~Meyerson, and M.~Minkoff.
\newblock Approximation algorithms for orienteering and discounted-reward tsp.
\newblock {\em SIAM J. Comput.}, 37(2):653--670, May 2007.

\bibitem{BHW14}
P.~Bolzoni, S.~Helmer, K.~Wellenzohn, J.~Gamper, and P.~Andritsos.
\newblock Efficient itinerary planning with category constraints.
\newblock In {\em SIGSPATIAL/GIS'14}, pages 203--212, Dallas, Texas, 2014.

\bibitem{samap08}
L.~Castillo, E.~Armengol, E.~Onaind{\'\i}a, L.~Sebasti{\'a},
  J.~Gonz{\'a}lez-Boticario, A.~Rodr{\'\i}guez, S.~Fern{\'a}ndez, J.~D. Arias,
  and D.~Borrajo.
\newblock {SAMAP}: A user-oriented adaptive system for planning tourist visits.
\newblock {\em Expert Systems with Applications}, 34(2):1318--1332, 2008.

\bibitem{ChKoPa08}
C.~Chekuri, N.~Korula, and M.~P\'{a}l.
\newblock Improved algorithms for orienteering and related problems.
\newblock In {\em SODA'08}, pages 661--670, 2008.

\bibitem{ChPa05}
C.~Chekuri and M.~P{\'a}l.
\newblock A recursive greedy algorithm for walks in directed graphs.
\newblock In {\em FOCS'05}, pages 245--253, 2005.

\bibitem{Dunn92}
J.~S. Dunn.
\newblock Scheduling underway replenishment as a generalized orienteering
  problem.
\newblock Master's thesis, Naval Postgraduate School, Monterey, California,
  1992.

\bibitem{GKMP14}
D.~Gavalas, C.~Konstantopoulos, K.~Mastakas, and G.~Pantziou.
\newblock A survey on algorithmic approaches for solving tourist trip design
  problems.
\newblock {\em Journal of Heuristics}, 20(3):291--328, 2014.

\bibitem{GeLaSe98b}
M.~Gendreau, G.~Laporte, and F.~Semet.
\newblock A branch-and-cut algorithm for the undirected selective traveling
  salesman problem.
\newblock {\em Networks}, 32(4):263--273, 1998.

\bibitem{Kell89}
C.~Keller.
\newblock Algorithms to solve the orienteering problem: A comparison.
\newblock {\em European Journal of OR}, 41:224--231, 1989.

\bibitem{LiKuSm02}
Y.-C. Liang, S.~Kulturel-Konak, and A.~Smith.
\newblock Meta heuristics for the orienteering problem.
\newblock In {\em CEC '02}, pages 384 --389, 2002.

\bibitem{LuLiTs11}
E.~H.-C. Lu, C.-Y. Lin, and V.~S. Tseng.
\newblock Trip-mine: An efficient trip planning approach with travel time
  constraints.
\newblock In {\em MDM'11}, pages 152--161, 2011.

\bibitem{MSMY04b}
A.~Maruyama, N.~Shibata, Y.~Murata, K.~Yasumoto, and M.~Ito.
\newblock A personal tourism navigation system to support traveling multiple
  destinations with time restrictions.
\newblock In {\em 18th Int. Conf on Adv. Information Networking and
  Applications (AINA'04)}, pages 18--21, 2004.

\bibitem{RaYoKa92}
R.~Ramesh, Y.-S. Yoon, and M.~H. Karwan.
\newblock An optimal algorithm for the orienteering tour problem.
\newblock {\em INFORMS Journal on Computing}, 4(2):155--165, 1992.

\bibitem{RiTs13}
M.~N. Rice and V.~J. Tsotras.
\newblock Parameterized algorithms for generalized traveling salesman problems
  in road networks.
\newblock In {\em SIGSPATIAL/GIS'13}, pages 114--123, Orlando, Florida, 2013.

\bibitem{RiSa09}
G.~Righini and M.~Salani.
\newblock Decremental state space relaxation strategies and initialization
  heuristics for solving the orienteering problem with time windows with
  dynamic programming.
\newblock {\em Computers {\&} OR}, 36(4):1191--1203, 2009.

\bibitem{SeSe06}
Z.~Sevkli and F.~E. Sevilgen.
\newblock Variable neighborhood search for the orienteering problem.
\newblock In {\em ISCIS'06}, pages 134--143, 2006.

\bibitem{SKGKB07}
A.~Singh, A.~Krause, C.~Guestrin, W.~J. Kaiser, and M.~A. Batalin.
\newblock Efficient planning of informative paths for multiple robots.
\newblock In {\em IJCAI'07}, pages 2204--2211, 2007.

\bibitem{TaSm00}
F.~Tasgetiren and A.~Smith.
\newblock A genetic algorithm for the orienteering problem.
\newblock In {\em IEEE Congress on Evolutionary Computation}, 2000.

\bibitem{Tsiligrides84}
T.~A. Tsiligrides.
\newblock Heuristic methods applied to orienteering.
\newblock {\em J. Operation Research Society}, 35(9):797--809, 1984.

\bibitem{WSGJ95}
Q.~Wang, X.~Sun, B.~L. Golden, and J.~Jia.
\newblock Using artificial neural networks to solve the orienteering problem.
\newblock {\em Annals of OR}, 61:111--120, 1995.

\end{thebibliography}

\end{document}